%% file: main_Jiayu.tex
\documentclass{llncs}
\input{macros.tex}

\usepackage{float}
\usetikzlibrary{crypto.symbols}
\begin{document}
\title{Delegating Quantum Computation in the Quantum Random Oracle Model}
\author{Jiayu Zhang\inst{1}\thanks{Supported in part by NSF awards IIS-1447700 and AF-1763786}}
\institute{Boston University (\email{jyz16@bu.edu})}

\date{\today}
\maketitle
\begin{abstract}
	A delegation scheme allows a computationally weak client to use a server's resources to help it evaluate a complex circuit without leaking any information about the input (other than its length) to the server. In this paper, we consider delegation schemes for quantum circuits, where we try to minimize the quantum operations needed by the client. We construct a new scheme for delegating a large circuit family, which we call ``C+P circuits''. ``C+P'' circuits are the circuits composed of Toffoli gates and diagonal gates. Our scheme is non-interactive, requires small amount of quantum computation from the client (proportional to input length but independent of the circuit size), and can be proved secure in the quantum random oracle model, without relying on additional assumptions, such as the existence of fully homomorphic encryption. In practice the random oracle can be replaced by an appropriate hash function or block cipher, for example, SHA-3, AES.\par
	This protocol allows a client to delegate the most expensive part of some quantum algorithms, for example, Shor's algorithm. The previous protocols that are powerful enough to delegate Shor's algorithm require either many client side quantum operations or the existence of FHE. The protocol requires asymptotically fewer quantum gates on the client side compared to running Shor's algorithm locally.\par
	To hide the inputs, our scheme uses an encoding that maps one input qubit to multiple qubits. We then provide a novel generalization of classical garbled circuits (``reversible garbled circuits'') to allow the computation of Toffoli circuits on this encoding. We also give a technique that can support the computation of phase gates on this encoding.\par
	To prove the security of this protocol, we study key dependent message(KDM) security in the quantum random oracle model. KDM security was not previously studied in quantum settings.
\end{abstract}
\section{Introduction}
\footnotetext{The full version of this paper can be found at \url{https://arxiv.org/pdf/1810.05234.pdf}}
%Shor proposed a famous quantum algorithm in \cite{Shorsalg}, which allows us to factor big integers efficiently. This algorithm is important because it can be used to break the RSA cryptosystem.  Some other applications of quantum computation include quantum simulation\cite{QSim} and quantum machine learning\cite{QML}. Since \par
In computation delegation, there is a client holding secret data $\varphi$ and the description of circuit $C$ that it wants to apply, but it doesn't have the ability to compute $C(\varphi)$ itself. A delegation protocol allows the client to compute $C(\varphi)$ with the help from a more computationally powerful server. The delegation is \emph{private} if the server cannot learn anything about the input $\varphi$ during the protocol.  After some communications, the client can decrypt the response from the server and get the computation result (see Figure 1.) This problem is important in the quantum setting: it's likely that quantum computers, when they are built, will be expensive, and made available as a remote service. If a client wants to do some quantum computation on secret data, a quantum computation delegation protocol is needed.\par
\begin{figure}[H]
\centering
\begin{tikzpicture}
	\tikzstyle{BC} = [
    draw,
    rectangle,
    node distance=10pt,
    minimum width=2cm,
    minimum height=1cm,
    text width=2cm,
    align=center,
  ]
  \node[left] (ink) at (0,0) {\begin{tabular}{r}description of circuit $C$\\$\varphi$\end{tabular}};
  \node[BC] (ink3) at (2,0) {Client};
  \node[BC] (inks) at (5,0) {(Quantum) Server};
  \draw[edge] (ink) -- (ink3);
  \draw[edge] ([yshift=8pt] ink3.east) -- ([yshift=8pt] inks.west);
  \draw[edge] ([yshift=-8pt] inks.west) -- ([yshift=-8pt] ink3.east);
  \node (ink4) at (2,-1.3) {$C(\varphi)$};
  \draw[dashed, ->] (ink3) -- (ink4);
  \node[text width = 4cm, align=center] (ink5) at (5,-1.4) {Nothing about $\varphi$ can be retrieved (efficiently)};
  \draw[dashed, ->] (inks) -- (ink5);
\end{tikzpicture}
\caption{Delegation of (quantum) computation}
\end{figure}
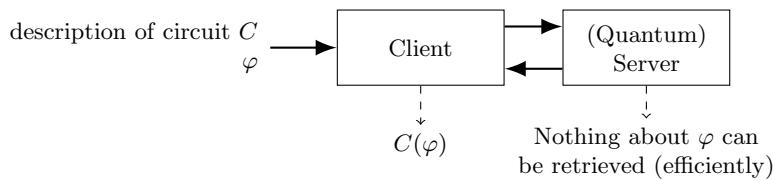

Delegation of computation is a central problem in modern cryptography, and has been studied for a long time in classical settings. Related works include multiparty computation, fully homomorphic encryption(FHE), etc. In the study of delegation, there are two key aspects: privacy and authenticity. This paper will focus on privacy.\par
%People have designed a lot of protocols for delegating quantum computation. \cite{UBQC} constructed a blind quantum computation protocol which is information-theoretically secure, but the interaction is heavy, and the client still needs to prepare a large quantum state in the beginning. For non-interactive protocols, \cite{AnnesQHELowT} introduced the idea of quantum homomorphic encryption(QHE), and constructed QHE schemes for low depth of $\fT$ gates, based on classical FHE. \cite{Mahadev2017ClassicalHE} constructed a protocol for universal quantum circuits, which only requires classical computing ability on the client side, as long as the inputs are classical. This protocol also requires classical FHE. What's more, \cite{QHEcode}\cite{statisticallyQHEforIQP} constructed information-theoretically secure QHE protocol, but only for some limited quantum circuit family. In section 1.2 we will introduce these work with more details.\par
We want the delegation protocol to be useful, efficient and secure. Previous work falls into two classes: some protocols have information-theoretical security, but they either can only support a small circuit class or require huge client side quantum resources (including quantum memories, quantum gates and quantum communications); other protocols rely on classical fully homomorphic encryption(FHE). This raises the following question:
\begin{center}
	\emph{Is it possible to delegate quantum computation for a large circuit family, with small amount of quantum resources on the client side, without assuming classical FHE?}
\end{center}
In the classical world, Yao's garbled circuit answers this question. Garbled circuit is also a fundamental tool in many other cryptographic tasks, like multiparty computation and functional encryption.
\paragraph{Note}When designing quantum cryptographic protocols, one factor that we care about is the ``quantum resources'' on the client side. The ``quantum resources'' can be defined as the sum of the cost of the following: (1)the size of quantum memory that the client needs; (2)the number of quantum gates that the client needs to apply; (3)the quantum communication that the client needs to make. Note that if the input (or computation, communication) is partly quantum and partly classical, we only consider the quantum part. Since the classical part is usually much easier to implement than the quantum part, as long as the classical part is polynomial, it's reasonable to ignore it and only consider the complexity of quantum resources. And we argue that it's better to consider the ``client side quantum resources'' instead of considering only the quantum memory size or quantum gates: on the one hand, we do not know which type of quantum computers will survive in the future, so it's better to focus on the cost estimate that is invariant to them; on the other hand, there may be some way to compose the protocol with other protocols to reduce the memory size, or simplify the gate set. %(One exception is the communication within the protocol: if the server needs to hold a big quantum state while the classical communication in the protocol "blocks" the quantum computation on the server side, it will become a burden.)
\subsection{Our Contributions}
In this paper we develop a non-interactive (1 round) quantum computation delegation scheme for ``C+P circuits'', the circuits composed of Toffoli gates and diagonal gates. We prove the following:\par
\begin{theorem}
It's possible to delegate C+P circuits non-interactively and securely in the quantum random oracle model, and the client requires $O(\eta N_q+N_q^2)$ quantum $\fCN$ gates as well as polynomial classical computation, where $N_q$ is the number of qubits in the input and $\eta$ is the security parameter.
\end{theorem}
We will give a more formal statement in Section 6. The client's quantum circuit size can in fact be bounded by $O(\kappa N_q)$ where $\kappa$ is the key length of the cryptographic primitives we use. Our current proof of security requires setting $\kappa = \eta + 4N_q$ where $\eta$ is the actual security parameter. However, we conjecture the same protocol can be proven secure for $\kappa =O(\eta)$, leading to the following conjecture:%The extra $N_q^2$ term comes from taking $\kappa=\eta+4N_q$ in $O(\kappa N_q)$ where $\kappa$ is the key length used in our protocol. Our current security proof techniques only work in this case, and we conjecture a better security proof exists when $\kappa=\eta$:
\begin{conjecture}
It's possible to delegate C+P circuits non-interactively and securely in the quantum random oracle model, using the same protocol as Theorem 1, and the client side quantum resources are $O(\eta N_q)$ $\fCN$ gates, where $N_q$ is the number of qubits in the input and $\eta$ is the security parameter.\end{conjecture}
%In our protocol, the client side quantum resources are actually  For the security proof, our current technique requires that , thus we get an $N_q^2$ term. We conjecture that this is not necessary(Conjecture \ref{conj:qindcpa}); and this implies:\par
%\begin{conjecture}
%	The client side quantum resources can be reduced to $O(\eta N_q)$, where $\eta$ is the security parameter.
%\end{conjecture}

We argue that our protocol is important for three reasons: (1)The client only needs small quantum resources. Here we say ``small'' to mean the quantum resources only depend on the key length and the input size, and are independent of the circuit size. (2)Its security can be proven in the quantum random oracle model, without assuming some trapdoor one-way function. Many protocols before, for example, \cite{QFHEPoly}\cite{Mahadev2017ClassicalHE} are based on classical FHE and therefore rely on some kinds of lattice cryptographic assumptions, for example, LWE assumption. Our protocol is based on the quantum random oracle (therefore based on hash functions in practice), and this provides an alternative, incomparable assumption on which we can base the security of quantum delegation. (3)Our protocol introduces some new ideas and different techniques, which may be useful in the study of other problems.\par
Our protocol can be applied to Shor's algorithm. The hardest part of Shor's algorithm is the Toffoli part applied on quantum states, so the client can use this protocol securely with the help of a remote quantum server. %\cite{ECCResource} estimates the complexity and gate counts for Shor's algorithm on elliptic curves. As the paper estimates, running Shor's algorithm on $n$ bit input requires $\Omega(n^3)$ Toffoli gates, while using our protocol, we have:
\begin{corollary}
It's possible to delegate Shor's algorithm on input of length $n$ within one round of communication in the quantum random oracle model, where the client requires ${O}(\eta n+n^2)$ $\fCN$ gates plus $\tilde O(n)$ quantum gates. Assuming Conjecture 1, the number of $\fCN$ gates is $O(\eta n)$. 
\end{corollary}
If the client runs the factoring algorithm by itself, the quantum operations it needed will be $\omega(n^2)$, and the exact complexity depends on the multiplication methods.\par
%Section 7.3 contains more details about it.\par
%Our protocol is closely related to the concept of key-dependent message(KDM) security. As a side product, we study KDM security in the quantum random oracle model. As far as we know, we are the first to study this problem. 
The security proof for our protocol heavily uses the concept of KDM security, which was not previously studied in the quantum setting. We therefore also initiate a systematic study of KDM security in the quantum random oracle model. We point out that although there already exists classical KDM secure encryption scheme in the random oracle model\cite{KDMCRO}, the security in the quantum random oracle model still needs an explicit proof. We complete its proof in this paper. Furthermore, we generalize KDM security to quantum KDM security,  and construct a protocol for it in the quantum random oracle model.
\subsection{Related Work}
To delegate quantum computation, people raised the concepts of blind quantum computation\cite{UBQC} and quantum homomorphic encryption(QHE)\cite{AnnesQHELowT}. These two concepts are a little different but closely related: in quantum homomorphic encryption, no interaction is allowed and the circuits to be evaluated are known by the server. While in blind quantum computation, interactions are usually allowed and the circuits are usually only known by the client.\par
The concept of blind quantum computation was first raised in \cite{BlindQC}. And \cite{UBQC} gave a universal blind quantum computation protocol, based on measurement-based quantum computation(MBQC)\cite{MBQC}. What's more, secure assisted quantum computation based on quantum one-time pad(QOTP) technique was raised in \cite{ChildsQOTP}, with which we can easily apply Clifford gates securely but $\fT$ gates are hard to implement and require interactions.\par
Quantum homomorphic encryption is the homomorphic encryption for quantum circuits. Based on QOTP and classical FHE, \cite{AnnesQHELowT} studied the quantum homomorphic encryption for circuits with low $\fT$ gate complexity. Later \cite{QFHEPoly} constructed a quantum homomorphic encryption scheme for polynomial size circuits. But it still requires some quantum computing ability on the client side to prepare the evaluation gadgets, and the size of gadgets is propotional to the number of $\fT$ gates. Recently Mahadev constructed a protocol\cite{Mahadev2017ClassicalHE}, which achieves fully quantum homomorphic encryption, and what makes this protocol amazing is that the client can be purely classical, which hugely reduces the burden on the client side.\par
Another viewpoint of these protocols is the computational assumptions needed. With interactions, we can do blind quantum computation for universal quantum circuits information theoretically(IT-) securely. But for non-interactive protocols, \cite{ITQHELimit} gave a limit for IT-secure QHE, which implies IT-secure quantum FHE is impossible. But it's still possible to design protocols for some non-universal circuit families.\cite{statisticallyQHEforIQP} gave a protocol for IQP circuits, and \cite{QHEcode} gave a protocol for circuit with logarithmic number of $\fT$ gates.\par
%\begin{figure}\includegraphics[scale=0.5]{fiveworlds2.png}\caption{Relations among different assumptions}\end{figure}
On the other hand, \cite{AnnesQHELowT}\cite{QFHEPoly}\cite{Mahadev2017ClassicalHE} rely on classical FHE. The current constructions of classical FHE are all based on various kinds of lattice-based cryptosystems, and the most standard assumption is the Learning-With-Error(LWE) assumption.\par
Table 1 compares different protocols for quantum computation delegation.
\begin{table*}
\begin{tabular}{cm{8em}cm{5em}}
Protocol&Circuit class&Client's quantum resources&Assumption\\\hline
QOTP\cite{ChildsQOTP}&Clifford&$O(N_q)$ Pauli operations&-\\\hline
\cite{UBQC}&All&\makecell{${O(L)}$\\Rounds: Circuit Depth}&-\\\hline
\cite{Mahadev2017ClassicalHE}&All&$O(N_q)$ Pauli operations&FHE\\\hline
\cite{statisticallyQHEforIQP}&IQP&$O(N_q)$&-\\\hline
\cite{QHEcode}&Clifford+small number of $\fT$ gates&Exponential in the number of $\fT$ gates&-\\\hline
This paper&C+P&\makecell{$O(\eta N_q)$(Conjectured)\\$O(\eta N_q+N_q^2)$(Proved)\\$\fCN$ operations}&Quantum ROM\\\hline
\end{tabular}
\caption{$L$ is the number of gates in the circuits, $N_q$ is the number of qubits in the input, $\eta$ is the security parameter.}
\end{table*}
\subsection{Techniques}
\subsubsection*{A different encoding for hiding quantum states with classical keys}
In many previous protocols, the client hides a quantum state using ``quantum one time pad'': $\rho\rightarrow \fX^a\fZ^b(\rho)$, where $a,b$ are two classical strings. After taking average on $a,b$, the encrypted state becomes a completely mixed state. In our protocol, we use the following mapping to hide quantum states, which maps one qubit in the plaintext to $\kappa$ qubits in the ciphertext:
$$\fEt_{k_0,k_1}: \ket{0}\rightarrow \ket{k_0},\ket{1}\rightarrow \ket{k_1}$$
where $k_0,k_1$ are chosen uniformly at random in $\{0,1\}^\kappa$ and distinct.\par
We can prove for all possible input states, if we apply this operator on each qubit, after taking average on all the possible keys, the final results will be exponentially close to the completely mixed state.
\subsubsection*{Reversible garbled circuits}
The main ingredient in our construction is ``reversible garbled circuit''. In the usual construction of Yao's garbled table, the server can feed the input keys into the garbled table, and get the output keys; then in the decoding phase, it uses an output mapping to map the keys to the result. This well-studied classical construction does not work for quantum states. Even if the original circuit is reversible, the evaluation of Yao's garbled circuit is not! To use it on quantum states, besides the original garbled table, we add another table from the output keys to the input keys. This makes the whole scheme reversible, which means we can use it on quantum states and the computation result won't be entangled with auxiliary qubits. For security, we remove the output mappings. In the context of delegation, these are kept by the client.\par
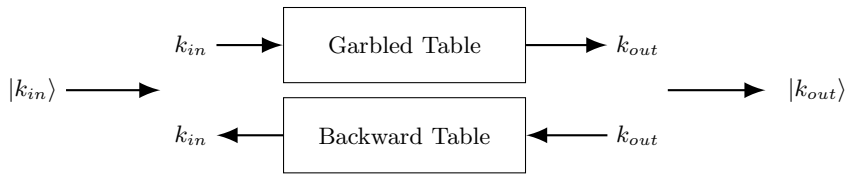
\begin{figure}[H]
\centering
%\begin{subfigure}{.3\textwidth}
%\begin{tikzpicture}
%\centering
%\tikzstyle{BC} = [
%    draw,
%    rectangle,
%    node distance=10pt,
%    minimum width=2cm,
%    minimum height=1cm,
%    text width=2cm,
%    align=center,
%  ]
%  \node[left] (ink) at (0,0) {$k_{in}$};
%	\node[BC] (gc) at (2,0) {Garbled Table};
%	\node[BC] (om) at (7,0) {Output Mapping};
%	\node[right] (res) at (9,0) {Result};
%	\node[left] (outk) at (5,0) {$k_{out}$};
%	\draw[edge] (ink) -- (gc);
%	\draw[edge] (gc) -- (outk);
%	\draw[edge] (outk) -- (om);
%	\draw[edge] (om) -- (res);
%\end{tikzpicture}
%\caption{The usual garbled table}
%\end{subfigure}
%\hspace*{\fill}
%\begin{subfigure}{.3\textwidth}
\begin{tikzpicture}
\tikzstyle{BC} = [
    draw,
    rectangle,
    node distance=10pt,
    minimum width=3cm,
    minimum height=1cm,
    text width=3cm,
    align=center,
  ]
  \node[left] (rink) at (0,-0.6) {$\ket{k_{in}}$};
  \node[left] (ink) at (2,0) {$k_{in}$};
	\node[BC] (gc) at (4.5,0) {Garbled Table};
	\node[left] (outk) at (8,0) {$k_{out}$};
	\node[left] (bink) at (2,-1.2) {$k_{in}$};
	\node[BC] (bgc) at (4.5,-1.2) {Backward Table};
	\node[left] (boutk) at (8,-1.2) {$k_{out}$};
	\node[left] (routk) at (10.5,-0.6) {$\ket{k_{out}}$};
\draw[edge] (ink) -- (gc);
	\draw[edge] (gc) -- (outk);
	\draw[edge] (boutk) -- (bgc);
	\draw[edge] (bgc) -- (bink);
	\draw[edge] (rink) -- ++(1.7,0);
	\draw[edge] (8,-0.6) -- ++(1.3,0);
\end{tikzpicture}
\caption{Reversible garbled table}
%\end{subfigure}
\end{figure}
\paragraph{Note}The proof of security of this scheme is subtle. The extra information included to allow the reversible computation introduces encryption cycles among the keys. We address the problem by studying key-dependent message security in the quantum setting. We show that a KDM-secure encryption scheme exists in the quantum random oracle model, and use this result to prove the security of our reversible garbled circuit construction. 
\subsubsection*{Phase gates}
The reversible garbled circuit allows evaluating Toffoli circuits. To handle phase gates, instead of applying $\ket{k_{in}}\rightarrow \ket{k_{out}}$, we can make the garbled table implement the following transformation(where $m$ is chosen randomly):
\begin{equation}
	\ket{k_0}\rightarrow\ket{k_0}\ket{m}, \ket{k_1}\rightarrow\ket{k_1}\ket{m+1}\label{eq:1}
\end{equation}
Then the server can apply a ``qudit $\fZ$ gate'' $\sum_i\omega^i_n\ket{i}\bra{i}$ (define $\omega_n:=e^{\mi\pi/n}$) on the second register, where $i\in \bZ_n$ goes through all the integers in $\bZ_n$.(This operation can be done efficiently.) This will give us:
$$\ket{k_0}\rightarrow\omega^m_n\ket{k_0}\ket{m}, \ket{k_1}\rightarrow\omega^{m+1}_n\ket{k_1}\ket{m+1}$$
Then it applies (\ref{eq:1}) again to erase the second register. After removing the global phase the result is the same as the output of applying a phase gate $R_Z(\frac{\pi}{n})=\ket{0}\bra{0}+\omega_n\ket{1}\bra{1}$.
\subsection{Organisation}
This paper is organized as follows. Section 2 contains some background for this paper%, including quantum computation, the definitions of several cryptography schemes and their security definitions, a simple review of garbled table in classical world, and a simple introduction of the quantum random oracle model
. In Section 3 we discuss the encoding scheme. In Section 4 we give our construction of the quantum computation delegation protocol for C+P circuits. In Section 5 we prove the security of classical KDM secure scheme in the quantum random oracle model, as the preparation for the security proof of the main protocol. Then in Section 6 we discuss the security of our protocol. Section 7.1 turns this delegation scheme to a fully blind protocol, and Section 7.2 shows how to use our protocol on Shor's algorithm. Section 8 generalizes KDM security to quantum settings, constructs a quantum KDM secure protocol and proves its security. Then we discuss the open questions and complete this paper.% The appendix contains some omitted proofs.
\section{Definitions and Preliminaries}
\subsection{Basics of Quantum Computation}
In this section we give a simple introduction for quantum computing, and clarify some notations in this paper. For more detailed explanations, we refer to \cite{NielsenChuangs}.\par
In quantum computing, a pure state is described by a unit vector in a Hilbert space. A qubit, or a quantum bit, in a pure state, can be described by a vector $\ket{\varphi}\in \bC^2$. The symbols $\ket{\cdot}$ and $\bra{\cdot}$ are called Dirac symbols. A qudit is described by a vector $\ket{\varphi}\in \bC^d$.\par
But a quantum system isn't necessarily in a pure state. When the quantum system is open, we need to consider mixed states. To describe both pure and mixed states, the state of a qubit is described by a density matrix in $\bC^{2\times 2}$. A density matrix is a trace-one positive semidefinite complex matrix. The density matrix that corresponds to pure state $\ket{\varphi}$ is $\ket{\varphi}\bra{\varphi}$, and we abbreviate it as $\varphi$.\par
For an n-qubit state, its density matrix is in $\bC^{2^n\times 2^n}$. The space of density operators in system $\cS$ is denoted as $\bD(\cS)$. Note that we use $\bE$ for the notation of the expectation value.\par
A quantum operation on pure states can be described by a unitary transform $\ket{\varphi}\rightarrow U\ket{\varphi}$. And an operation on mixed states can be described by a superoperator $\rho\rightarrow \cE(\rho)=\tr_R(U(\rho\otimes\ket{0}\bra{0})U^\dagger))$. We use calligraphic characters like $\cD$, $\cE$ to denote superoperators, and use the normal characters like $U,D$ to denote unitary transforms. We also use Sans-serif font like $\fX,\fZ,\fEt$ to denote quantum operations: When they are used as $\fEt\ket{\varphi}$ they mean unitary operations(applied on Dirac symbols without parentheses), and when used as $\fEt(\rho)$ they mean superoperators.\par
The quantum gates include $\fX$, $\fY$, $\fZ$, $\fCN$, $\fH$, $\fT$, $\fToffoli$ and so on. What's more, denote $R_Z(\theta)=\ket{0}\bra{0}+e^{\mi\theta}\ket{1}\bra{1}$, where $\mi$ is the imaginary unit. Denote $\omega_n=e^{\mi\pi/n}$, we can write $R_Z(k\pi/n)=\ket{0}\bra{0}+\omega_n^k\ket{1}\bra{1}$. Since the $i$ will be used as the symbol for indexes and ``inputs'', we avoid using $e^{\mi\pi/n}$ in this paper, and use $\omega_n$ instead.\par
% A gate is a diagonal gate if in its matrix representation all the non-zero terms are on its diagonal.\par
The trace distance of two quantum states is defined as $\Delta(\rho,\sigma)=\frac{1}{2}|\rho-\sigma|_{\tr}$, where $|\cdot|_{\tr}$ is the trace norm.
\subsection{Encryption with Quantum Adversaries}
A quantum symmetric key encryption scheme contains three mappings: $\fKg(1^\kappa)\rightarrow sk$, $\fEn_{sk}: \bD(\cM)\rightarrow \bD(\cC)$, $\fDc_{sk}: \bD(\cC)\rightarrow \bD(\cM)$.\cite{QPKC}\par
In this paper, we need to use the symmetric key encryption scheme with key tags, which contains four mappings: $\fKg$, $\fEn$, $\fDc$, $\fVer$. The scheme has a key verification procedure $\fVer: \cK\times \bD(\cC)\rightarrow \{\perp, 1\}$.\par
A quantum symmetric key encryption scheme with key tags is correct if:
\begin{enumerate}
	\item $\forall \rho\in \bD(\cR\otimes \cS)$, $\bE_{sk\leftarrow \fKg(1^\kappa)}|(\fI\otimes\fDc_{sk})((\fI\otimes\fEn_{sk})(\rho))-\rho|_{\tr}=\fneg(\kappa)$
	\item $\forall \rho\in \bD(\cR\otimes \cS)$, $\Pr_{sk\leftarrow \fKg(1^\kappa)}(\fVer(sk, (\fI\otimes\fEn_{sk})(\rho))=\perp)=\fneg(\kappa)$,\\
	      and $\Pr_{sk\leftarrow \fKg(1^\kappa), r\leftarrow \fKg(1^\kappa)}(\fVer(r, (\fI\otimes\fEn_{sk})(\rho))=1)=\fneg(\kappa)$
\end{enumerate}
Here the encryption and decryption are all on system $S$, and $R$ is the reference system.\par
Sometimes we also need to encrypt the messages with multiple keys, and require that (informally) an adversary can only get the message if it knows all the keys. In symmetric multi-key encryption scheme with key tags, $\fKg(1^\kappa)$ is the same as the symmetric single-key scheme, $\fEn_{k_1,k_2,\cdots k_i}$ encrypts a message under keys $K=(k_1,k_2,\cdots k_i)$, $\fDc_{k_1,k_2,\cdots k_i}$ decrypts a ciphertext given all the keys $k_1,k_2,\cdots k_i$, and $\fVer(k,i,c)\rightarrow \{\perp, 1\}$ verifies whether $k$ is the $i$-th key used in the encryption of $c$.\par
The next problem is to define ``secure'' formally. The concept of indistinguishability under chosen plaintext attack (IND-CPA) was introduced in \cite{SymINDCPA}\cite{PubINDCPA}. Let's first review the security definitions in the classical case.
\begin{definition}
	For a symmetric key encryption scheme, consider the following game, called ``IND-CPA game'', between a challenger and an adversary $\sA$:
	\begin{enumerate}
		\item The challenger runs $\fKg(1^\kappa)\rightarrow sk$ and samples $b\leftarrow_r\{0,1\}$.
		\item The adversary gets the following classical oracle, whose input space is $\cM$:%, and tries it for some times:
		      \begin{enumerate}
			      \item The adversary chooses $m\in \cM$, and sends it into the oracle.
			      \item If $b=1$, the oracle outputs $\fEn(m)$.  If $b=0$, it outputs $\fEn(0^{|m|})$.
		      \end{enumerate}
		\item The adversary tries to guess $b$ with some distinguisher $\cD$. Denote the guessing result as $b^\prime$.
	\end{enumerate}
	The distinguishing advantage is defined by $\fAdv^{IND-CPA}(\sA,\kappa)=|\Pr(b^\prime=1|b=1)-\Pr(b^\prime=1|b=0)|$.\par
	And we call it an one-shot IND-CPA game if the adversary can only query the oracle once. Similarly we can define the distinguishing advantage\\ $\fAdv^{IND-CPA-oneshot}(\sA,\kappa)=|\Pr(b^\prime=1|b=1)-\Pr(b^\prime=1|b=0)|$

\end{definition}
\begin{definition}
	We say a protocol is IND-CPA secure against quantum adversaries if for any BQP adversary $\sA$ which can run quantum circuits as the distinguisher but can only make classical encryption queries, there exists a negligible function $\fneg$ such that $\fAdv^{IND-CPA}(\sA,\kappa)=\fneg(\kappa)$. And we call it one-shot IND-CPA secure against quantum adversaries if $\fAdv^{IND-CPA-oneshot}(\sA,\kappa)=\fneg(\kappa)$.
\end{definition}
Note that the ``IND-CPA security against quantum adversaries'' characterizes the security of a protocol against an adversary who has the quantum computing ability in the distinguishing phase but can only run the protocol classically.\par
For quantum cryptographic schemes, we use the formulation in \cite{AnnesQHELowT}.
\begin{definition}
	For a symmetric key encryption scheme, consider the following game, called ``qIND-CPA game'', between a challenger and an adversary $\sA$:
	\begin{enumerate}
		\item The challenger runs $\fKg(1^\kappa)\rightarrow sk$ and samples $b\leftarrow_r\{0,1\}$.
		\item The adversary gets the following oracle, whose input space is $\cD(\cM)$:
		      \begin{enumerate}
			      \item The adversary chooses $\rho\in \bD(\cM\otimes \cR)$. The adversary sends system $\cM$ to the oracle, and keeps $\cR$ as the reference system.
			      \item If $b=1$, the oracle applies $\fEn$ on $\cM$ and sends it to the adversary. The adversary will hold the state $(\fEn\otimes \fI)(\rho)$.  If $b=0$, the oracle encrypts $0^{|m|}$ and the adversary gets $(\fEn\otimes \fI)(0^{|m|}\otimes\rho_R)$, where $\rho_R$ is the density operator of subsystem $\cR$.
		      \end{enumerate}
		\item The adversary tries to guess $b$ with some distinguisher $\cD$. Denote the guessing output as $b^\prime$.
	\end{enumerate}
	The distinguishing advantage is defined by $\fAdv^{qIND-CPA}(\sA,\kappa)=|\Pr(b^\prime=1|b=1)-\Pr(b^\prime=1|b=0)|$.\par
	And we call it an one-shot qIND-CPA game if the adversary can only query the oracle once. Similarly we can define the distinguishing advantage\\ $\fAdv^{qIND-CPA-oneshot}(\sA,\kappa)=|\Pr(b^\prime=1|b=1)-\Pr(b^\prime=1|b=0)|$.
\end{definition}
\begin{definition}
	A protocol is qIND-CPA secure if for any BQP adversary $\sA$, there exists a negligible function $\fneg$ such that $\fAdv^{qIND-CPA}(\sA,\kappa)=\fneg(\kappa)$.\par
	What's more, we call it one-shot qIND-CPA secure if for any BQP adversary $\sA$, there exists a negligible function $\fneg$ such that $\fAdv^{qIND-CPA-oneshot}(\sA,\kappa)=\fneg(\kappa)$.
\end{definition}
In the definition of qIND-CPA security, the adversary can query the encryption oracle with quantum states, and it can also run a quantum distinguisher.
\paragraph{Key Dependent Message security} In the definitions above the plaintexts do not depend on the secret keys. There is another type of security called ``key-dependent message (KDM) security'', where the adversary can get encryptions of the secret keys themselves. We will need to study this type of security in the proof of our main theorem, but we defer the definitions and further discussions to Section 5.
\subsection{Delegation of Quantum Computation, and Related Problems}
There are three similar concepts: delegation of quantum computation, quantum homomorphic encryption\cite{AnnesQHELowT} and blind quantum computation\cite{BlindQC}\cite{UBQC}.\par
The differences of these three concepts are whether the interaction is allowed, and which party knows the circuit. The delegation of quantum computation and blind quantum computation protocols are interactive. For quantum homomorphic encryption, the interaction is not allowed. If we focus on non-interactive protocols, their difference is which party knows the circuit: in blind quantum computation, the circuit is only known by the client but not the server; in homomorphic encryption, the circuit is known by the server but not necessarily known by the client. In our paper, we use ``delegation of quantum computation'' to mean that the circuit is known by both parties but the input is kept secret.\par
A non-interactive quantum computation delegation protocol $\fBQC$ on circuit family $\cF=\{F_n\}$ contains 4 mappings:
\begin{description}
	\item $\fBQC.\fKg(1^\kappa,1^N,1^L)\rightarrow (sk)$: The key generation algorithm takes the key length $\kappa$, input length $N$ and circuit length $L$ and returns the secret key.
	\item $\fBQC.\fEn^C_{sk}:\bD(\cM)\rightarrow \bD(\cC)$. Given the encryption key and the public circuit in $\cF=\cup\{F_n\}$, this algorithm maps inputs to ciphertexts.
	\item $\fBQC.\fEv^C:\bD(\cC)\rightarrow \bD(\cC^\prime)$. This algorithm maps ciphertexts to some other ciphertexts, following the instructions which may be contained in $\cC$.
	\item $\fBQC.\fDc_{sk}:\bD(\cC^\prime)\rightarrow \bD(\cM^\prime)$. This algorithm decrypts the ciphertexts and stores the outputs in $\cM$.
\end{description}
%And we can define a weak form of blind quantum computation, I will call it "semi-blind quantum computation", where the input is hidden but the circuit is public:\par
%\begin{description}
%	\item $\fsBQC.\fEv^C:D(\cC)\rightarrow D(\cC^\prime)$, we write it in this way to show $C$ is public.
%\end{description}
%The definition of quantum homomorphic encryption protocol is raised in \cite{AnnesQHELowT}. Here we revise it a little bit:\par
%\begin{description}
%	\item $\fQHE.\fKg(1^\kappa,1^N,1^L)\rightarrow (sk)$: The key generation operation takes security parameter and returns public key, secret key, and evaluation key.
%	\item $\fQHE.\fEn^C_{pk}:\bD(\cM)\rightarrow \bD(\cC)$. Given public key, maps inputs to a ciphertext.
%	\item $\fQHE.\fEv^C_{pk,evk}:\bD(\cC)\rightarrow \bD(\cC^\prime)$. Evaluate a circuit in $\cF$ on ciphertext in $\cC$.
%	\item $\fQHE.\fDc_{sk}:\bD(\cC^\prime)\rightarrow \bD(\cM)$. Same as before.
%\end{description}
Here we put $N,L$ into the $\fKg$ algorithm, which are needed in our protocol. We put $C$ on the superscript to mean the circuit is known by both parties.\par
\begin{definition}
	The security (IND-CPA, qIND-CPA, etc) of the non-interactive delegation of computation protocol is defined to be the security of its encryption scheme $(\fKg,\fEn)$.
\end{definition}
\subsection{Quantum Random Oracle Model}
A classical random oracle is an oracle of a random function $\cH:\{0,1\}^\kappa\rightarrow \{0,1\}^\kappa$ which all parties can query with classical inputs. It returns independent random value for different inputs, and returns fixed value for the same input. In practice, a random oracle is usually replaced by a hash function.\par
A quantum random oracle allows the users to query it with quantum states: the users can apply the map $\cH: \ket{a}\ket{b}\rightarrow \ket{a}\ket{\cH(a)\oplus b}$ on its state. The quantum random oracle was raised in \cite{QRO}. It becomes the security proof model for many post-quantum cryptographic scheme\cite{FS}. On the other hand, the application of the quantum random oracle in quantum cryptographic problems is not very common, and as far as we know, our work is the first application of it in the delegation-stype problems.\par
The security definitions in the quantum random oracle model are the same as Definitions 2 and 4. Here we assume the adversary can only make polynomial number of random oracle queries, but the queries can be quantum states. Then by the ``Random Oracle Methodology'' we can conjecture the protocol is also secure in the standard model, when the random oracle is replaced by a hash function in practice. As with proofs in the classical random oracle model, interpreting these security claims is subtle, since there exist protocols that are secure in the random oracle model but insecure in any concrete initialization of hash function.\cite{CGH-ROLimit}\par
%Classical symmetric key KDM secure encryption scheme can be constructed in classical random oracle model\cite{KDMCRO}. But its security in quantum random oracle model has not been studied yet. In this paper we give such a proof.\par
This paper focuses on the quantum cryptographic protocols in the quantum random oracle model. As far as we know, the assumption of a quantum random oracle is incomparable to any trapdoor assumption. We do not know any construction of public key encryption based on solely quantum random oracle. What's more, in our proof, the random oracle doesn't need to be ``programmable''\cite{ROProg}.
\subsection{Garbled Table}\label{clgarble}
We make a simple introduction of Yao's garbled table \cite{YaoGCOrigin} here. The garbled table construction will be the foundation of our protocol.\par
Garbled table is a powerful technique for the randomized encoding of functions. When constructing the garbled circuit of some circuit $C$, the client picks two keys for each wire, and denotes them as $k_b^w$, where $b\in \{0,1\}$, and $w$ is the index of the wire.\par
The garbled table is based on a symmetric key encryption scheme with key tags. For gate $g$, suppose its input wires are $w_1, w_2$, and the output wire is $v$. The client constructs the following table:
\begin{align}
	  & \fEn_{k_0^{w_1},k_0^{w_2}}(k_{g(0,0)}^v) \\
	  & \fEn_{k_0^{w_1},k_1^{w_2}}(k_{g(0,1)}^v) \\
	  & \fEn_{k_1^{w_1},k_0^{w_2}}(k_{g(1,0)}^v) \\
	  & \fEn_{k_1^{w_1},k_1^{w_2}}(k_{g(1,1)}^v)
\end{align}
And it picks a random permutation in $S_4$ to shuffle them.\par
If the server is given the garbled table for some gate, and given a pair of input keys, it can evaluate the output keys: it can try each row in the garbled table and see whether the given keys pass the verification. If they pass, use them to decrypt this row and get the output keys.\par
By providing the input keys and the garbled table for each gate in the circuit, the server can evaluate the output keys for the whole circuit. And in the randomized encoding problem the client also provides the mapping from the output keys to the corresponding values on some wires: $k_b^w\rightarrow b$, for some set of $w$s. The server can know the output values on these revealed wires, but the values on other wires are hidden. This construction has wide applications in classical world, for example, it allows an $NC^0$ client to delegate the evaluation of a circuit to the server.
%After computing the garbled table for the whole circuit, we encrypt each bit in the input with two keys. If the bit is 0, we use $k_0$ for this bit; if it's 1, we use $k_1$. And we can see this is just a special case of our entanglement encryption scheme! So Instead of giving the corresponding input keys to the server, we use $\fEt$ to encode the input, with the garbled circuit input keys being the randomness of $\fEt$. We can see when the input is classical, this is just the garbled circuit construction.

\section{The Encoding For Hiding Quantum States With Classical Keys}
Let's first discuss the encoding operator, $\fEt$,  to ``hide'' the quantum states. For each qubit in the input, the client picks two random different keys $k_0,k_1\in\{0,1\}^\kappa$ and encodes the input qubit with the following operator:
$$\fEt_{k_0,k_1}: \ket{0}\rightarrow \ket{k_0},\ket{1}\rightarrow \ket{k_1}$$
The dimensions of two sides are not the same, but we can add some auxiliary qubits on the left side. As long as $k_0, k_1$ are distinct, this operator is unitary.\par
For pure quantum state $\ket{\varphi}=\sum\alpha_{i_1i_2\cdots i_N}\ket{i_1i_2\cdots i_N}$, given key set $K=\{k^{n}_i\}$, where $n\in [N]$, $i\in \{0,1\}$, if we apply this operator on each qubit, using keys $\{k^n_0,k^n_1\}$ for the $n$-th qubit, we get:
$$\fEt_{K}\ket{\varphi}=\sum\alpha_{i_1i_2\cdots i_n}\ket{k^{(1)}_{i_1}k^{(2)}_{i_2}\cdots k^{(N)}_{i_n}}$$
The following lemma shows that if the keys are long enough, chosen randomly and kept secret, this encoding is statistically secure, in other words, the mixed state after we take average on all the possible keys, is close to the completely mixed state with exponentially small distance:
\begin{lemma}\label{lm:itecd}
	Suppose $\rho\in \bD(\cS\otimes \cR)$, $\cS=(\mathbb{C}^2)^{\otimes N}$. Suppose we apply the $\fEt$ operation on system $\cS$ with key length $\kappa$, after taking average on all the valid keys, we get
	$$\sigma=\frac{1}{(2^\kappa(2^\kappa-1))^N}\sum_{\forall n\in [N], k_0^n,k_1^n\in \{0,1\}^\kappa, k_0^n\neq k^n_1}(\fEt^\cS_{\{k^{n}_i\}}\otimes\fI)(\rho)$$
	then we have $\Delta(\sigma,(\frac{1}{2^{\kappa N}}\fI)\otimes \tr_{\cS}(\rho))\leq(\frac{1}{2})^{\kappa-4}N$
\end{lemma}
Thus such an encoding keeps the input secure against unbounded adversaries. We put the detailed proof in the full version of this paper.\par
Since $\fEt$ is a unitary mapping, given $K$ and $\fEt_K(\rho)$, we can apply the inverse of $\fEt$ and get $\rho$: $\fEt_K^{-1}(\fEt_K(\rho))=\rho$. Note that when applying $\fEt$ we enlarge the space by appending auxiliary qubits, and when applying $\fEt^{-1}$ we remove these auxiliary qubits.\par

\begin{fact}
	$\fEt$ can be implemented with only $\fCN$ operations.
\end{fact}
\begin{proof}
First implement mapping $\ket{0}\rightarrow \ket{0^\kappa},\ket{1}\rightarrow\ket{k_0\oplus k_1}$. This can be done by $\fCN$ the input into the places where $k_0\oplus k_1$ has bit value 1. Then apply $\fX$ gates on the places where $k_0$ has bit value 1. This will xor $k_0$ into these registers and complete the mapping $\ket{0}\rightarrow \ket{k_0},\ket{1}\rightarrow\ket{k_1}$.
\end{proof}
The quantum computation delegation protocol that we will discuss in the next section will use this encoding.
\section{A Quantum Computation Delegation Protocol for C+P Circuits}
In this section, we use $\fEt$ encoding and a new technique called ``reversible garbled circuit'' to design a quantum computation delegation protocol.
\subsection{C+P Circuits and the Relation to Toffoli Depth}
\cite{ClassicalTD1} defined ``almost classical'' circuits. Here we rename it to ``C+P'' circuits, abbreviating ``classical plus phase''.
\begin{definition}[\cite{ClassicalTD1}]
	C+P is the family of quantum circuits which are composed of Toffoli gates and diagonal gates.
\end{definition}
%The family of C+P circuits is a little similar to IQP+ circuits. IQP+ is introduced in \cite{statisticallyQHEforIQP}, which contains $\fCN$ and diagonal gates. Another difference from \cite{statisticallyQHEforIQP} is we don't put restriction on the input and output. The input could be any quantum state and there is no additional measurement in the end.\par
We can prove it's possible to decompose this type of circuits into simpler gates. We put the proof in the full version of this paper.
\begin{proposition}\label{fact:C+Pdecomp}
	Any C+P circuit can be decomposed to Toffoli gates and single qubit phase gates. Furthermore, it can be approximated by Toffoli gates and single qubit phase gates of the form $R_Z(\frac{\pi}{n})=\ket{0}\bra{0}+\omega_n\ket{1}\bra{1}, n\in \bN_+$, where $\omega_n$ is the $nth$ root of unity. To approximate a circuit of length $L$ of Toffoli gates and single qubit phase gates to precision $\epsilon$, we only need Toffoli gates and phase gates in the form of $R_Z(\frac{\pi}{2^d}), d\in [D]$, where $D=\Theta(\log\frac{L}{\epsilon})$.
\end{proposition}
We consider $D$ as a fixed value in this paper. Since $\epsilon$ depends exponentially on $D$, a small $D$ in practice should be enough and it will at most add a logarithmic term in the complexity.\par
$\{$C+P, $\fH\}$ is a complete basis for quantum circuits. Our work implies a delegation scheme whose round complexity equals the H-depth of a given circuit. Previous works on quantum computation delegation generally focused on $\{$Clifford, $\fT\}$ basis. (The exception is \cite{statisticallyQHEforIQP}, which works for IQP circuits.) With the exception of Mahadev's FHE-based scheme\cite{Mahadev2017ClassicalHE}, their complexity of client side quantum gates increases with the circuit's T-depth.\par
As far as we know, there is no general way to transform a Toffoli circuit into the $\{$Clifford, $\fT\}$ basis such that its $\fT$ depth is smaller than the Toffoli depth of the original ciruit, without blowing up the circuit width exponentially. We formalize this statement as a conjecture:
\begin{conjecture}
	For any polynomial time algorithm that transforms Toffoli circuits into the $\{$Clifford, $\fT\}$ basis, there exists a sequence of inputs with increasing Toffoli depths for which the algorithm's outputs have $\fT$ depth $\Omega(d)$, where $d$ denotes the Toffoli depths of the original circuits.
\end{conjecture}
Working with the $\{$C+P, $\fH\}$ basis allows us to design efficient protocols for delegating Shor's algorithm (which has low H-depth). Previously, this was only possible using FHE-based schemes.
%When we design protocols on $\{$Clifford, $\fT\}$ basis, the $\fT$ gates are usually hard to implement. In this paper, we study the delegation problem on the $\{$C+P, $\fH\}$ basis, and the $\fH$ gates are hard to implement. However, some famous algorithms(for example, Shor's algorithm), have very small $\fH$ depth. Our protocol can perform well on these problems, compared to many other protocols on $\{$Clifford, $\fT\}$ basis.
\subsection{Protocol Construction}\label{sec:prtlintro}
We now describe our protocol that supports our main results. This protocol gets a public description of a C+P circuit as well as a secret quantum state.\par
The idea comes from Yao's Garbled Circuit construction. We have discussed the construction in section \ref{clgarble}. The garbled circuit construction is commonly used for randomized encodings of classical circuits, but it's not applicable to quantum circuits. In this paper we will show how to do the reversible garbling for C+P circuits. Let's first discuss the ideas briefly.\par
One big difference of classical operations and quantum operations is in quantum world, the operations have to be reversible. Firstly, we will consider the garbling of Toffoli gates. In classical world, the garbled tables can contain non-reversible gates, for example, AND gate, OR gate. But in quantum world, we have to start with the Toffoli gate, which is reversible, and contains 3 input wires and 3 output wires.\par
However, even if the underlying circuit is reversible, if we try to use the classical garbled table construction on a quantum circuit, the garbled circuits is still not reversible, and it's not possible to use it to implement the quantum operations. Note that we need two levels of reversibility here: the circuit to be garbled needs to be reversible, and the garbled circuit itself has to be reversible too, even if it calls the random oracle as a black box.\par
Thus we propose a new garbling technique, which is a reversible garbling of reversible circuits: when constructing the garbled tables, instead of just creating one table for each gate, the client can construct two tables, in one table it encrypts the output keys with the input keys, and in the other table it encrypts the input keys with the output keys! This construction will make the garbled circuit reversible: we will show, the garbled circuit evaluation mapping can be applied on quantum states unitarily.\par
But another problem arises: If we simply replace the garbled circuit in the randomized encoding problem with ``reversible garbled circuit'', it's not secure any more. But it turns out, if we remove the output mapping, it becomes secure again, under some reasonable assumptions. And that gives us a delegation protocol.\par
The full protocol is specified in Protocol \ref{prtl:gbc}. Below we give more details.
\subsubsection{Reversible garbling of Toffoli gates}
First recall that in the classical garbled circuit, the evaluation operation on each garbled gate takes the input keys, decrypts the table and computes the corresponding output keys:\par
$$k_{in}\rightarrow k_{out}$$
This mapping is classical, and there is a standard way to transform a classical circuit to a quantum circuit, by introducing auxiliary output registers, and keeping the input:
\begin{equation}
	U:\ket{k_{in}}\ket{c}\xrightarrow{\text{garbled gate}}\ket{k_{in}}\ket{k_{out}\oplus c}\label{eq:5}
\end{equation}
We use the second register as the output register, and $c$ is its original value. This mapping computes the output keys from the garbled table and xors them to the second register.\par
This mapping is unitary, and we can also put superpositions on the left-hand side of (\ref{eq:5}). However, when it is used directly on quantum states, the inputs and outputs will be entangled together. Explicitly, for a specific Toffoli gate, we use $k^{w1}_u,k^{w2}_v, k^{w3}_w$ to denote the keys of the input wires $w1,w2,w3$ which correspond to the input $(u, v, w)$; for the output part the keys are $k^{v1}_u,k^{v2}_v, k^{v3}_w$. If we apply (\ref{eq:5}) directly, we get:
\begin{align*}U:&\ket{k^{w1}_u}\ket{k^{w2}_v}\ket{k^{w3}_w}\ket{c_1}\ket{c_2}\ket{c_3}\\&\rightarrow\ket{k^{w1}_u}\ket{k^{w2}_v}\ket{k^{w3}_w}\ket{k^{v1}_u\oplus c_1}\ket{k^{v2}_v\oplus c_2}\ket{k^{v3}_{w\oplus uv}\oplus c_3}\end{align*}
But what we need is the following mapping:
\begin{equation}U:\ket{k^{w1}_u}\ket{k^{w2}_v}\ket{k^{w3}_w}\rightarrow\ket{k^{v1}_u}\ket{k^{v2}_v}\ket{k^{v3}_{w\oplus uv}}\label{eq:101}\end{equation}
Which means, we need to disentangle and erase the input registers from the output registers. Note that, again, both sides should be understood as superpositions of different keys. And recall that for each Toffoli gate there are eight possible combinations of input keys, and this mapping should work for all the eight combinations.\par
To erase the input from the output, we can use two mappings: $\ket{k_{in}}\ket{0}\rightarrow\ket{k_{in}}\ket{k_{out}}$ and $\ket{k_{in}}\ket{k_{out}}\rightarrow\ket{0}\ket{k_{out}}$. Both operations have the same form as equation (\ref{eq:5}). (For the second step, we could view the $k_{out}$ as the input, $k_{in}$ as $c$, and get $\ket{k_{out}}\ket{k_{in}\oplus k_{in}}$) So we can use two garbled tables for this ``reversible garbled table''!\par

%So instead of computing $\fCL.\fEn_{k_{in}}(k_{out})$ as in the classical garbled table, for each gate and each input-output pair, the client computes:
%$$\fCL.\fEn_{k_{in}}(k_{out})\qquad\qquad \fCL.\fEn_{k_{out}}(k_{in})$$
Assume $\fCL$ is some multiple key encryption scheme with key tags. The client puts the encryption outputs $\fCL.\fEn_{k_{in}}(k_{out})$ into a table(there are eight rows in this table), and shuffles them randomly; this is the forward table. And it puts the encryption outputs $\fCL.\fEn_{k_{out}}(k_{in})$ into a table and shuffles to get the backward table. This construction will allow the server to implement (\ref{eq:101}), even on superpositions of input keys.\par%The server will use the first table to apply $\ket{k_{in}}\ket{0}\rightarrow\ket{k_{in}}\ket{k_{out}}$, and use the second table to compute the inputs from the output keys: when it's applied on $\ket{k_{in}}\ket{k_{out}}$, it "erases" the inputs and implements the mapping $\ket{k_{in}}\ket{k_{out}}\rightarrow\ket{0}\ket{k_{out}}$.\par
We note that we do not need to consider the detailed operations for decrypting each garbled table, and the existence of such operations comes from quantize the classical mapping as (\ref{eq:5}).\par
For the encoding of the inputs, recall that in the usual garble table construction, the client encrypts each bit in the inputs with the mapping:
 \begin{equation}
 	0\rightarrow k_0,1\rightarrow k_1\label{eq:6}
 \end{equation}
To make it quantum, instead of replacing the classical bits with the corresponding keys, the client uses $\fEt$ operator to hide the inputs. And we notice that (\ref{eq:6}) is a special case of $\fEt$ where the input is classical.\par
%\paragraph{Note}We didn't discuss the details of the evaluation operation, and the reader may get confused on whether the auxiliary information will get entangled with the original system. We note that the correctness is argued in an abstract way, and there is actually no need to dive into the details--which might make the whole process even harder to understand.
%So the evaluation can be abstracted as follows: there is an efficient evaluation algorithm $\cO$ such that $\cO(\fEt^S_{K_{in}}(\rho),\fTAB^C(K))=\fEt^S_{K_{out}}(C(\rho))$, where $\fTAB^C(K)$ is the "reversible garbled table" of keys, and just revealing $\fTAB^C(K)$ won't reveal the keys so the plaintext will be hidden.
\subsubsection{Phase gates}
Now the protocol works for Toffoli gates. But what if there are phase gates?\par
From Proposition \ref{fact:C+Pdecomp}, we only need to consider the single qubit phase gates in the form of $R_Z(\frac{\pi}{n}), n\in \bZ_+$. Suppose we want to implement such a gate on some wire, where the keys are $k_0,k_1$, corresponding to values $0$ and $1$, as discussed in the last subsection.\par
To implement $R_Z(\frac{\pi}{n})$, the client first picks a random integer $m\in \bZ_n$. What it is going to do is to create a table of two rows, put $\fCL.\fEn_{k_0}(m)$ and $\fCL.\fEn_{k_1}(m+1)$ into the table and shuffle it. When the server needs to evaluate $R_Z(\frac{\pi}{n})$, it will first decrypt the garbled table and write the output on an auxiliary register $\ket{0}$. So it can implement the following transformation:
\begin{equation}
	\ket{k_0}\rightarrow\ket{k_0}\ket{m}, \ket{k_1}\rightarrow\ket{k_1}\ket{m+1}\label{eq:7}
\end{equation}
This step is similar to implementing equation (\ref{eq:5}).\par
Then it applies a ``qudit $Z$ gate'' $\sum_i\omega^i_n\ket{i}\bra{i}$ on the second register, where $i\in \bZ_n$ goes through all the integers in $\bZ_n$.(This operation can be done efficiently.) This will give us:
$$\ket{k_0}\rightarrow\omega^m_n\ket{k_0}\ket{m}, \ket{k_1}\rightarrow\omega^{m+1}_n\ket{k_1}\ket{m+1}$$
Then it applies (\ref{eq:7}) again to erase the second register. After removing the global phase the result is the same as the output of applying a phase gate $R_Z(\frac{\pi}{n})=\ket{0}\bra{0}+\omega_n\ket{1}\bra{1}$.\par
What's more, since $m$ is chosen randomly the garbled gate won't reveal the keys. (This fact is contained in the security proof.)

\subsection{Protocol Design}

In this section we formalize this garbled circuit based quantum computation delegation protocol. Let's call it $\fGBC$.\par
We index the wires in the circuit as follows: If two wires are separated by a single qubit phase gate, we consider them as the same wire; otherwise (separated by a Toffoli gate, or disjoint), they are different wires. Suppose we have already transformed the circuit using Fact \ref{fact:C+Pdecomp} so that there is no controlled phase gate. For a circuit with $N$ input bits and $L$ gates, the number of wires is at most $N+3L$.
\begin{prtl}\label{prtl:gbc} The protocol $\fGBC$, with $\fCL$ being the underlying classical encryption scheme, for a circuit $C$ which is composed of Toffoli gates and phase gates in the form of $R_Z(\frac{\pi}{n})$, is defined as:
	\begin{description}
		\item[Key Generation]$\fGBC.\fKg(1^\kappa,1^N,1^L)$: Sample keys $K=(k_b^l)$,\\ $k_b^l\leftarrow \fCL.\fKg(1^\kappa)$, $b\in \{0,1\},l\in [N+3L]$.
		\item[Encryption]$\fGBC.\fEn_K^C(\rho)$:
		Output $(\fEt_{K_{in}}(\rho), \fTAB_\fCL^C(K))$. (Note that with the reference system, the first part is $(\fI\otimes \fEt_{K_{in}})(\rho_{RS})$.)
		\item[{Evaluation}] $\fGBC.\fEv^C(c)$, where $c=(\rho_q,tabs)$: Output $\fEvTAB_\fCL^C(\rho_q,tabs)$
		\item[{Decryption}]$\fGBC.\fDc_K(\sigma)$: Suppose the output keys in $K$ are $K_{out}$. Apply the map $\fEt^{-1}_{K_{out}}(\cdot)$ on $\sigma$ and return the result.
	\end{description}
\end{prtl}
$\fTAB_\fCL^C(K)$ and $\fEvTAB_\fCL^C(\rho_q,tabs)$ appeared in this protocol are defined as follows:
\begin{prtl} $\fTAB_{\fCL}^C(K)$, where $K$ is the set of keys:\par
	Suppose circuit $C$ is composed of gates $(g_i)_{i=1}^L$. This algorithm returns $(tab_{g_i})_{i=1}^L$, where $tab_g$ is defined as follows:
	\begin{enumerate}
		\item If $g$ is a Toffoli gate: Suppose $g$ has controlled input wires $w1,w2$ and target wire $w3$, and the corresponding output wires are $v1,v2,v3$. Suppose the corresponding keys in $K$ are $\{k^w_b\},w\in \{w1,w2,w3,v1,v2,v3\},b\in\{0,1\}$:\par
Create $table1$ as follows: For each triple $u,v, w \in \{0,1\}^3$, add the following as a row:
$$\fCL.\fEn_{k_u^{w1},k_v^{w2},k_w^{w3}}(k_u^{v_1}||k_v^{v_2}||k_{w\oplus uv}^{v_3})$$
			            and pick a random permutation in $S_8$ to shuffle this table.\par
			            Create $table2$ as follows: For each triple $u,v, w \in \{0,1\}^3$, add the following as a row:
			            $$\fCL.\fEn_{k_u^{v_1},k_v^{v_2},k_{w\oplus uv}^{v_3}}(k_u^{w1}||k_v^{w2}||k_w^{w3})$$
			            and pick another random permutation in $S_8$ to shuffle this table.\\
			       Return $(table1,table2)$

		\item If $g$ is a phase gate, Suppose $g$ is a phase gate $R_Z(\frac{\pi}{n})$ on wire $w$:\par
		Sample $m_0\leftarrow_r \bZ_n$, $m_1=m_0+1$. Create $table1$ as follows: For each $u \in \{0,1\}$, add the following as a row:
			            $$\fCL.\fEn_{k_u^{w}}(m_{u})$$
			            and pick a random permutation in $S_2$ to shuffle this table.\\ Return $table1$.
	\end{enumerate}
\end{prtl}
\begin{prtl} $\fEvTAB^C_{\fCL}(\rho, tab)$:\par
	Suppose circuit $C$ is composed of gates $(g_i)_{i=1}^L$. For each gate $g$ in $C$, whose corresponding garbled gate is $tab_g$ in $tab$:\par
	If $g$ is a Toffoli gate, with input wires $w1,w2,w3$, output wires $v1,v2,v3$: Suppose $tab_g=(tab1,tab2)$, where $tab1$ is the table from input keys to output keys, and $tab2$ is from output keys to input keys. Suppose $\rho\in \bD(\cS_g\otimes \cS^\prime)$, where $\cS_g$ is the system that is currently storing the keys on the input wires of $g$, and $\cS^\prime$ is the remaining systems:
	\begin{enumerate}
		\item Introduce three auxiliary registers and denote the system as $S_g^\prime$. Use $tab1$ to apply the following mapping on $S_g$, as discussed in the section \ref{sec:prtlintro}:
		$$\ket{k^{w1}_u}\ket{k^{w2}_v}\ket{k^{w3}_w}\ket{0}\ket{0}\ket{0}\rightarrow\ket{k^{w1}_u}\ket{k^{w2}_v}\ket{k^{w3}_w}\ket{k^{v1}_u}\ket{k^{v2}_v}\ket{k^{v3}_{w\oplus uv}}$$
		\item Use $tab2$ to apply the following mapping on $\cS_g\otimes \cS^\prime_g$, as discussed in the section \ref{sec:prtlintro}:
		$$\ket{k^{w1}_u}\ket{k^{w2}_v}\ket{k^{w3}_w}\ket{k^{v1}_u}\ket{k^{v2}_v}\ket{k^{v3}_{w\oplus uv}}\rightarrow\ket{0}\ket{0}\ket{0}\ket{k^{v1}_u}\ket{k^{v2}_v}\ket{k^{v3}_{w\oplus uv}}$$
		\item Remove system $\cS_g$, rename $\cS_g^\prime$ as $\cS_g$. Denote the final state as the new $\rho$.
	\end{enumerate}
	If $g$ is a phase gate on wire $w$ in the form of $R_Z(\frac{\pi}{n})$, : Suppose $\rho\in D(\cS_g\otimes \cS^\prime)$, where $\cS_g$ is the system that stores the keys on the input wire of $g$, and $S^\prime$ is the remaining systems:
	\begin{enumerate}
		\item Use $tab_g$ to implement the mapping $\ket{k^w}\ket{0}\rightarrow\ket{k^w}\ket{m}$, where $m$ is the decrypted output.
		\item Apply $\sum_i\omega^i_n\ket{i}\bra{i}$ on the system of $m$.
		\item Use $tab_g$ to implement the mapping $\ket{k^w}\ket{m}\rightarrow\ket{k^w}\ket{0}$.
	\end{enumerate}
\end{prtl}
The following two theorems summarize its correctness and efficiency:
\begin{theorem}
	Protocol $\fGBC$ is a correct non-interactive quantum computation delegation protocol for C+P circuits.
\end{theorem}
%\subsubsection{Efficiency}
\begin{theorem}\label{thm:complexity}
In $\fGBC$ protocol, the quantum resources required on the client side are  $O(\kappa N_q)$ $\fCN$ gates, where $\kappa$ stands for the key length used in the protocol, $N_q$ is the size of quantum states in the input, which are independent of the size of the circuit.	
\end{theorem}
Here we use $N_q$ instead of $N$ because we want to consider the case where some part of the input is classical and some part of it is quantum. To make the protocol secure we may need to choose $\kappa$ depending on $N_q$. This is discussed with more details in Section \ref{sec:security}.\par
This means the quantum resources of this protocol are independent of the circuit to be evaluated! In practice the size of the circuit may be a large polynomial of the input size, and our protocol will not be affected by this.
%\subsection{On the Related Impossibility Results of Quantum Computation Delegation}
%There are several impossibility results on the delegation of quantum computation, for example, \cite{ITQHELimit}\cite{AaronsonLimit}. \cite{ITQHELimit} gave us a limit on IT-secure QHE, and their result can be applied to C+P gates. That might make us think that delegation of quantum computation on C+P gates must rely on trapdoor one-way functions, but our protocol bypasses the limit in two ways: (1) making the circuit public; (2) using the quantum random oracle. In practice, these conditions don't affect the usability too much; but they lead to something different, and lots of open questions arise here. (See Section 9 for further discussion.)
\subsection{Structure of the Security Proofs}
The structure of the security proofs is as follows. First we study the key dependent message security in the quantum world, and design a protocol which we call the $\fKDMP$ protocol. Note that this part is not about the garbling scheme.\par
Then for the garbling scheme, we first state Proposition 2, which is the IND-CPA security of our garbling scheme. And we state a lemma about the security of the garbling scheme, which is the Lemma 2. The proofs use a reduction to the security of the $\fKDMP$ protocol. And the proofs are in the full version.\par
Then we prove the security of our garbling scheme (Theorem 6) from Proposition 2 and lemma 2. This part is given in the main content.
\section{KDM Security of Classical Encryption against Quantum Attack}
As we can see, in $\fGBC$ protocol there are encryption cycles. So to make the protocol secure, for the underlying encryption scheme $\fCL$, the usual security definition is not enough and we need at least KDM security. In this section, we will first discuss the key dependent message security(KDM security) in quantum world, and give an encryption scheme $\fKDMP$ that is KDM-secure against quantum adversaries. These results will be the foundation for the security proof of the $\fGBC$ protocol.\par
In classical world, KDM security was discussed in several papers, for example, \cite{KDMCRO,cliqueKDMFromLWE}. \cite{KDMCRO} gave a classical KDM secure encryption scheme in the random oracle model, and \cite{cliqueKDMFromLWE} constructed KDM secure protocols in the standard model, based on some hard problems, for example, Learning-With-Error.\par
\subsection{KDM Security in the Classical World}
As a part of the preliminaries, we repeat the definition of the security game of the classical KDM security\cite{KDMCRO}\cite{cliqueKDMFromLWE}.\par
\begin{definition}
	The KDM-CPA game is defined similar to the IND-CPA game, except that (1)in the first step the challenger runs $\fKg(1^\kappa)$ for $N$ times to generate $\cK=\{sk_i\}_{i\in [N]}$, $N$ is less than a polynomial of the security parameter.(2)the client is allowed to query the encryption oracle with a function $f\in \cF$, a message $m$, and an index $i$ of the keys, and the encryption oracle returns $\fEn_{sk_i}(f(K,m))$ or $\fEn_{sk_i}(0^{|f(K,m)|})$, depending on $b$. Note that the outputs of functions in $\cF$ should be fixed-length, otherwise $|f(K,m)|$ is not well-defined.
\end{definition}
%A classical KDM secure encryption scheme was constructed and proven secure in the classical random oracle model in \cite{KDMCRO}.\par

\subsection{KDM Security in the Quantum World}
The attack for the KDM security can be adaptive, which means, the adversary can make encryption queries after it receives some ciphertexts. But in our work we only need to consider the non-adaptive setting. What's more, we only need to consider the symmetric key case. To summarize, the game between the adversary and the challenger can be defined as:
\begin{definition}[naSymKDM Game]
	%This definition is parameterized by a family $\cF$ of classical functions.
	
 The symmetric key non-adaptive KDM game naSymKDM for function family $\cF$ against a quantum adversary $\sA$ in the quantum random oracle model with parameters $(\kappa, L, T, q)$ is defined as follows.
	\begin{enumerate}
		\item The challenger chooses bit $b\leftarrow_r\{0,1\}$ and samples $K=\{sk_i\}_{i=1}^{L}$, $sk_i\leftarrow\fKg(1^\kappa)$.
		\item The adversary and the challenger do the following $T$ times, non-adaptively, which means, the challenger will only send out the answers in step (b) after it has received all the queries:
		      \begin{enumerate}
			      \item The adversary picks index $i$, function $f\in \cF$ and message $msg\in\{0,1\}^*$, and sends them to the challenger. The size of $msg$ should be compatible with $f$.
			      \item If $b=1$, the challenger gives $c=\fEn_{sk_i}(f(K,msg))$ to the adversary. If $b=0$, the challenger gives $c=\fEn_{sk_i}(0^{|f(K,msg)|})$.
		      \end{enumerate}

		\item The adversary tries to guess $b$ using distinguisher $\cD$ and outputs $b^\prime$. Here $\cD$ is a quantum operation and can query the oracle with quantum states. Suppose $\cD$ will query the random oracle for at most $q$ times.
	\end{enumerate}
$f$ can also query the random oracle, and it only makes queries on classical states. What's more, the output of functions in $\cF$ should have a fixed length, otherwise $|f(K,m)|$ will not be well-defined.\par
	The guessing advantage is defined as $\fAdv^{naSymKDM}_{\cF}(\sA_{(L,T,q)},\kappa)=|\Pr(b^\prime=1|b=1)-\Pr(b^\prime=1|b=0)|$.
\end{definition}
\begin{definition}
	A symmetric key encryption scheme is nonadaptive KDM secure for circuit family $\cF$ against quantum adversaries in the quantum random oracle model if for any $BQP$ adversary, $$\fAdv^{naSymKDM}_{\cF}(\sA_{(L(\kappa),T(\kappa),q(\kappa))},\kappa)=\fneg(\kappa)$$
	Where $L(\kappa),T(\kappa),q(\kappa)$ are polynomial functions that may depend on the adversary.
\end{definition}

\subsection{A KDM Secure Protocol in the Quantum Random Oracle Model}
In the quantum random oracle model, we can give a construction of the classical KDM secure encryption scheme $\fKDMP$. Here ``classical'' means the encryption and decryption are purely classical. But the distinguisher may query the quantum random oracle in superposition.
\begin{prtl}\label{prtl:kdmro} We can construct a symmetric KDM secure encryption scheme $\fKDMP$ that has key tags in the quantum random oracle model, where we denote the random oracle as $\cH$:
	\begin{description}
		\item $\fKDMP.\fKg(1^\kappa)$: Output $sk\leftarrow_r \{0,1\}^\kappa$
		\item $\fKDMP.\fEn_{sk}(m)$: $R_1,R_2\leftarrow_r\{0,1\}^\kappa$, output ciphertext $c=(R_1,\cH(sk||R_1)\oplus m)$ and key tag $(R_2, \cH(sk||R_2))$\item $\fKDMP.\fDc_{sk}(c)$: Output $\cH(sk||c_1)\oplus c_2$, where $c_1$ and $c_2$ are from $c=(c_1,c_2)$.
		\item $\fKDMP.\fVer(k,tag)$: Suppose $tag=(tag1,tag2)$, output $1$ if $\cH(k||tag1)=tag2$, and $\perp$ otherwise.
	\end{description}
\end{prtl}
Since the execution of this protocol is classical, the correctness can be proved classically and is obvious. We refer to \cite{KDMCRO} here and write it out explicitly for convenience.
\begin{theorem}[Correctness]
	$\fKDMP$ is a correct symmetric key encryption scheme with key tags in the quantum random oracle model.
\end{theorem}
The security under classical random oracle model has been proven. But here we study the quantum random oracle, so although the protocol is almost the same, we still need a new proof.
\begin{theorem}[Security]\label{thm:rosec}
	Define $\cF[q^\prime]$ as the set of classical functions that query the random oracle at most $q^\prime$ times. For any adversary which can query the random oracle quantumly at most $q$ times, we have
	$$\fAdv_{\fKDMP, \cF[q^\prime]}^{naSymKDM}(\sA_{(L,T,q)},\kappa)\leq \fpoly(q,q^\prime,L,T)2^{-0.5\kappa}$$
	where $\fpoly$ is a fixed polynomial.
\end{theorem}
We put the proof in the full version of this paper.
\section{Security of $\fGBC$ Protocol}\label{sec:security}
In this section we discuss the security of protocol $\fGBC$. First we need to construct a classical encryption scheme $\fCL$ as its underlying scheme. The construction is very similar to the $\fKDMP$ scheme, except that this is multi-key and the $\fKDMP$ scheme is single-key. We will use it as the underlying scheme of $\fGBC$.
\subsection{Construction of the Underlying Classical Encryption Scheme}

\begin{prtl}\label{prtl:clro} The underlying multi-key encryption scheme $\fCL$ is defined as:
	\begin{description}
		\item $\fCL.\fKg(1^\kappa)$: Output $sk\leftarrow_r \{0,1\}^\kappa$
		\item $\fCL.\fEn_{k_1,k_2,k_3}(m)$: $R_1,R_2,R_3,R_4,R_5,R_6\leftarrow_r\{0,1\}^\kappa$, output
		\begin{gather}
			(R_1,R_2,R_3,\cH(k_1||R_1)\oplus\cH(k_2||R_2)\oplus\cH(k_3||R_3)\oplus m),\\
			((R_4, \cH(k_1||R_4)), (R_5, \cH(k_2||R_5)), (R_6, \cH(k_3||R_6)))
		\end{gather}

		where $\cH$ is the quantum random oracle.
		\item $\fCL.\fDc_{k_1,k_2,k_3}(c)$: Suppose $c=(R_1,R_2,R_3,c_4)$. Output $(\cH(k_1||R_1)\oplus\cH(k_2||R_2)\oplus\cH(k_3||R_3)\oplus c_4)$.
		\item $\fCL.\fVer(k,i,c)$: Suppose the $i$th key tag in $c$ is $tag_i=(R_i,r)$. Output $1$ if $r=\cH(k||R_i)$, and $\perp$ otherwise.
	\end{description}
\end{prtl}
We choose not to define and discuss the security of this scheme, but use it as a ``wrapper'' of the $\fKDMP$ scheme. In the security proof we will ``unwrap'' its structure and base the proof on the security of $\fKDMP$ scheme.
\subsection{Security of $\fGBC$ against Classical or Quantum Attack}
In this subsection we give the security statements of $\fGBC$. First, we can show, when used on classical inputs, $\fGBC_\fCL$ is secure:
\begin{proposition}\label{prop:clsec}
	$\fGBC_\fCL$, where $\fCL$ is defined as Protocol \ref{prtl:clro}, is one-shot IND-CPA secure against quantum adversary (that is, secure when used to encrypt one \emph{classical} input) in the quantum random oracle model. Explicitly, if the distinguisher that the adversary uses makes at most $q$ queries to the quantum random oracle, the input size is $N$ and the size of circuit $C$ is $L$,
	$$\fAdv^{IND-CPA-oneshot}_{\fGBC_\fCL^C}(\sA,\kappa)\leq \fpoly(q,N,L)2^{-0.5\kappa}$$
	Where $\fpoly$ is a fixed polynomial that does not depend on $\sA$ or the parameters.
\end{proposition}
The detailed proof is in the full version of this paper.\par
But we meet some difficulty when we try to prove the qIND-CPA security (that is, the security for quantum inputs). We leave it as a conjecture:
\begin{conjecture}\label{conj:qindcpa}
	$\fGBC_\fCL$ is one-shot qIND-CPA secure in the quantum random oracle model.
\end{conjecture}

But if we use a longer key, we can prove its security.
\begin{theorem}\label{thm:qssec}
	%Suppose $N_q$ is the size of quantum states in the input, then $\fGBC_\fCL$ is one-shot qIND-CPA secure, in the quantum random oracle model when we take $\kappa\geq 4N_q$ and consider $\eta=\kappa-4N_q$ as the security parameter, where $\kappa$ is the key length. In other words, 
	For any BQP adversary $\sA$, there exists a negligible function $\fneg$ such that:
	$$\fAdv^{qIND-CPA-oneshot}_{\fGBC_\fCL}(\sA,\kappa)= \fneg(\kappa-4N_q)$$
	where $N_q$ is the size of quantum states in the input.
\end{theorem}
In other words, denote $\fGBC^\prime$ as the protocol of taking $\kappa=\eta+4N_q$ as the key length in the $\fGBC$ protocol, we can prove $\fGBC^\prime$ is one-shot qIND-CPA secure with respect to security parameter $\eta$. So we prove:\par
\begin{theorem}
There exists a delegation protocol for C+P gate set that is one-shot qIND-CPA secure in the quantum random oracle model, and the client requires $O(\eta N_q+N_q^2)$ quantum $\fCN$ gates as well as polynomial classical computation, where $N_q$ is the number of qubits in the input and $\eta$ is the security parameter.
\end{theorem}
Although we don't have a proof for Conjecture \ref{conj:qindcpa}, we conjecture it is true, since this protocol seems to be a very natural generalization from classical to quantum. We leave it as an open problem. The main obstacle here is its security cannot be reduced to the semantic security of classical garbled circuits easily: the adversary gets many superpositions of keys. We have to prove it using different techniques, which leads to Theorem \ref{thm:qssec}.\par
From Theorem \ref{thm:qssec} we know when we take $\kappa\geq 4N_q$ and consider $\kappa-4N_q$ as the security parameter the security has been proved. So when the circuit size $L=\omega(N^2_q)$ the quantum resources for the client to run this protocol are smaller than running the circuit itself anyway.\par
What's more, although our proof requires the quantum random oracle model, we conjecture that this protocol is still secure when we replace the random oracle with practical hash functions or symmetric key encryption schemes:
\begin{conjecture}\label{conj:qrotohash}
When we replace the quantum random oracle in $\fGBC_\fCL$ with practical hash functions or symmetric key encryption schemes, such as versions of SHA-3 or AES with appropriate input and output sizes, the security statements still hold.	
\end{conjecture}
%\paragraph{Note}
%The usage of random oracle in our protocol is closer to the symmetric key encryption. We use random oracle for both verification of the keys and the encryption, but we believe a point-and-permute technique for garbled circuit can be used here to replace the verification.

\subsection{Security Proof}
\subsubsection{IND-CPA security of Protocol \ref{prtl:gbc}}
The proof of Proposition \ref{prop:clsec} is postponed into the full version of this paper. The proof is based on Theorem \ref{thm:rosec}, which is about KDM security of Protocol \ref{prtl:kdmro}. The structure of our scheme, when used classically, can be seen as a special case of the KDM function. But the definition of IND-CPA security for protocol $\fGBC$ is still different from the KDM game security: in $\fGBC$ we are trying to say the inputs of $\fEt$ are hidden, but KDM security is about the encrypted messages in the garbled table. So it doesn't follow from the security of $\fKDMP$ protocol trivially.\par
\subsubsection{Discussions of the qIND-CPA security}
To prove Theorem \ref{thm:qssec}, we use a different security proof technique, which enables us to base the qIND-CPA advantage on the IND-CPA advantage and a classical ``hard-to-compute'' lemma. This technique enables us to argue about the security of a quantum protocol using only security results in the classical settings.\par
We need to prove the keys that are not ``revealed'' are ``hard to compute''. Then we expand the expression of the qIND-CPA advantage, write it as the sum of exponential number of terms and we can observe that their forms are the same as the probability of ``computing the unrevealed keys''. We can prove these terms are all exponentially small, thus we get a bound for the whole expression.\par
%So let's formalize the lemma.
\begin{lemma}\label{lem:mcluncp}
	For any C+P circuit $C$, $|C|=L$, any adversary that uses distinguisher $\cD$ which can query the quantum random oracle $q$ times (either with classical or quantum inputs), given the reversible garbled table and input keys corresponding to one input, it's hard to compute the input keys corresponding to other input. Formally, for any $i\neq j,\ket{\varphi_i}$, we have
	\begin{align}
&\bE_{K}\bE_R\tr((\fEt_K\ket{j})^\dagger\cD(\fEt_K(\ket{i}\bra{i})\otimes \varphi_i\otimes \fTAB^C_{\fCL}(K,R))(\fEt_K\ket{j}))\notag\\&\qquad\qquad\qquad\qquad\qquad\qquad\qquad\qquad\qquad\qquad\quad\leq \fpoly(q,N,L)2^{-0.5\kappa}\label{eq:25}
	\end{align}
	where $\fpoly$ is a fixed polynomial that does not depend on $\sA$ or the parameters, $N$ is the size of inputs, and $R$ denotes the randomness used in the computation of $\fTAB^C_\fCL(K)$, including the random oracle outputs, the random paddings and the random shuffling. And $\fTAB^C_{\fCL}(K,R)$ is the output of $\fTAB^C_\fCL(K)$ using randomness $R$, and since $R$ is given as a parameter there will be no randomness inside.

\end{lemma}
Note that since we have already fixed all the randomness, $\fTAB^C_{\fCL}(K,R)$ is pure. We also note that this can be seen as a classical lemma since $\ket{i}$, $\ket{j}$ are all in computational basis. We postpone the proof into the full version.\par
Let's prove Theorem \ref{thm:qssec} from Proposition \ref{prop:clsec} and Lemma \ref{lem:mcluncp}. We will expand the the expression of the input state and qIND-CPA advantage, and each term in the cross terms can be bounded by (\ref{eq:25}).%, whose distinguisher is defined as follows: suppose the distinguisher in the qIND-CPA game is $\cD$, first apply $\cD$, make a projection measurement, and inverse $\cD$. This will give us the bound we need.
\begin{proof}[of Theorem \ref{thm:qssec}]
	First, suppose the state that the adversary uses is $\ket{\varphi}=\sum_{i}c_i\ket{i}\ket{\varphi_i}$, where $i$ is in the input system, $i\in I$ where $I$ is the set of non-zero term($c_i\neq 0$), $|I|\leq  2^{N_q}$ and $\ket{\varphi_i}$ is in the reference system. Additionally assume $c_i$s are all real numbers and $|\ket{i}\ket{\varphi_i}|=1$. We can only consider pure states since we can always write a mixed state as a probability ensemble of pure states.\par
	Then we can assume the distinguisher $\cD$ is a unitary operation $D$ on the output and auxiliary qubits, followed by a measurement on a specific output qubit. So we can write $\cD(\rho)=\tr_R(D(\rho\otimes \ket{0}\bra{0})D^\dagger)$, where $\ket{0}\bra{0}$ stands for big enough auxiliary qubits. Let's use $\cE_{proj}(\rho)$ to denote the operation of projecting $\rho$ onto the computational basis. Denote the projection operator onto the $\ket{0}\bra{0}$ space as $P_0$, we have
	\begin{flalign}
		     & \fAdv^{qIND-CPA-oneshot}_\fGBC(\sA,\kappa)                                                                      \\
		=    & |\Pr(\cD(\bE_K\fGBC.\fEn_K(\varphi))=1))-\Pr(\cD(\bE_K\fGBC.\fEn_K(0^{N}))=1)|                          \\
		\leq & |\Pr(\cD(\bE_K\bE_R(\rho))=1))-\Pr(\cD(\bE_K\bE_R(\cE_{proj}(\rho)))=1)|+ \notag \\
		     & |\Pr(\cD(\bE_K\fGBC.\fEn_K(\cE_{proj}(\varphi)))=1))-\Pr(\cD(\bE_K\fGBC.\fEn_K(0^{N}))=1)|\label{eq:32}
	\end{flalign}
	Here we write $\rho:=(\fEt_K\otimes \fI)(\varphi)\otimes \fTAB(K,R)$.\par
	Let's first compute the first term.
	\begin{flalign}
		  & |\Pr(\cD(\bE_K\bE_R(\rho))=1))-\Pr(\cD(\bE_K\bE_R(\cE_{proj}(\rho)))=1))|                                                                          \\
		= & |\tr(P_0(\bE_K\bE_R D(\rho\otimes \ket{0}\bra{0})D^\dagger))-\tr(P_0(\bE_K\bE_R D(\cE_{proj}(\rho)\otimes \ket{0}\bra{0})D^\dagger))|\label{eq:38}
	\end{flalign}
	The first term inside can be expanded as
	\begin{align}
		  & \bE_K\bE_R D(\rho\otimes \ket{0}\bra{0})D^\dagger                                                                                                                                                          \\
		= & \bE_{K}\bE_R D((\fEt_K\otimes \fI)(\varphi)\otimes \fTAB(K,R)\otimes \ket{0}\bra{0})D^\dagger                                                                                                               \\
		= & \bE_{K}\bE_R D((\fEt_K\otimes\fI)((\sum_{i}c_i\ket{i}\ket{\varphi_i})(\sum_{i}c^\dagger_i\bra{i}\bra{\varphi_i}))\notag\\
		&\qquad\qquad\qquad\qquad\qquad(\fEt_K\otimes\fI)^\dagger\otimes \fTAB(K,R)\otimes \ket{0}\bra{0})D^\dagger\label{eq:35}
	\end{align}
	Denote $\ket{x_i}=\fEt_K\ket{i}\otimes \ket{\varphi_i}$, we can simplify the expression:
	\begin{align}
		(\ref{eq:35})= & \bE_{K}\bE_R D(\sum_{i}c_i\ket{x_i}\sum_{i}c^\dagger_i\bra{x_i}\otimes \fTAB(K,R)\otimes \ket{0}\bra{0})D^\dagger \\
		=              & \bE_{K}\bE_R D(\sum_{i}|c_i|^2\ket{x_i}\bra{x_i}\otimes \fTAB(K,R)\otimes \ket{0}\bra{0})D^\dagger     \notag           \\
		               & +\bE_{K}\bE_R D(\sum_{i\neq j}c_ic_j^\dagger\ket{x_i}\bra{x_j}\otimes \fTAB(K,R)\otimes \ket{0}\bra{0})D^\dagger  \\
		=  & \bE_K\bE_R D(\cE_{proj}(\rho)\otimes \ket{0}\bra{0})D^\dagger                                                      \notag\\
		 & +\bE_{K}\bE_RD(\sum_{i\neq j}c_ic^\dagger_j\ket{x_i}\bra{x_j}\otimes \fTAB(K,R)\otimes \ket{0}\bra{0})D^\dagger
	\end{align}
	Substitute it into (\ref{eq:38}), we get
	\begin{flalign}
		&(\ref{eq:38})\notag\\
		= & |\bE_K\bE_R\tr(P_0D(\sum_{i\neq j}c_ic^\dagger_j\ket{x_i}\bra{x_j}\otimes \fTAB(K,R)\otimes \ket{0}\bra{0})D^\dagger)|                                                \label{eq:24n}\\
		=              & |\sum_{i\neq j}c_ic^\dagger_j\bE_K\bE_R(\bra{x_j}\bra{\fTAB(K,R)}\bra{0} D^\dagger P_0 D (\ket{x_i}\ket{\fTAB(K,R)}\ket{0})| \label{eq:25n}                                         \\
		\leq           & \sqrt{\sum_{i\neq j}c_i^2{c^\dagger_j}^2}\sqrt{\sum_{i\neq j}|\bE_K\bE_R\bra{0}\bra{\fTAB(K,R)}\bra{x_j} D^\dagger P_0 D \ket{x_i}\otimes \ket{\fTAB(K,R)}\ket{0}|^2} \\
		\leq           & \sqrt{\sum_{i\neq j}\bE_K\bE_R|(\bra{0}\otimes\bra{\fTAB(K,R)}\bra{x_j})D^\dagger P_0 D (\ket{x_i}\otimes \ket{\fTAB(K,R)}\ket{0})|^2}\label{eq:48}
	\end{flalign}
The magic of this technique actually happens between (\ref{eq:24n}) and (\ref{eq:25n}): first we move $\sum_{i\neq j}c_ic^\dagger_j$ out by linearity, then after rotating terms inside the trace, an expression which talks about applying $D$ on some state becomes an expression for the probability of applying $\{D^\dagger P_0D,D^\dagger P_1D\}$ on $\ket{x_i}$ and getting $\ket{x_j}$.\par

	By Lemma \ref{lem:mcluncp}, consider the operation $\cE$ defined as follows: expand the space and apply $D$, make a measurement with operators $\{P_0,P_1\}$, and apply $D^\dagger$. Let $\cE_0=D^\dagger P_0D(\cdot\otimes \ket{0}\bra{0})D^\dagger P_0D$, and $\cE_1=D^\dagger P_1D(\cdot\otimes \ket{0}\bra{0})D^\dagger P_1D$. We have:
	\begin{align}
		     & \bE_K\bE_R(\tr((\fEt_K\ket{j})^\dagger \cE_0(\fEt_K(i)\otimes \varphi_i\otimes \fTAB(K,R))\fEt_K\ket{j}))    \\
		     & +\tr((\fEt_K\ket{j})^\dagger \cE_1 (\fEt_K(\ket{i}\bra{i})\otimes \varphi_i\otimes\fTAB(K,R))\fEt_K\ket{j})) \\
		\leq & \fpoly(q,N,L)2^{-0.5\kappa}
	\end{align}
	With this, we can bound the inner part of (\ref{eq:48}) further:
	\begin{align}
		     & \bE_K\bE_R|(\bra{0}\otimes\bra{\fTAB(K,R)}\bra{x_j})D^\dagger P_0 D (\ket{x_i}\otimes \ket{\fTAB(K,R)}\ket{0})|^2                                                                   \\
		=    & \bE_K\bE_R|(\bra{0}\otimes\bra{\fTAB(K,R)}((\fEt_K\ket{j})\otimes \ket{\varphi_j})^\dagger \notag\\
		&\qquad \qquad \qquad \qquad\qquad D^\dagger P_0 D (\fEt_K\ket{i}\otimes \ket{\varphi_i})\otimes \ket{\fTAB(K,R)}\ket{0})|^2 \\
		\leq & \bE_K\bE_R\tr((\fEt_K\ket{j})^\dagger \cE_0(\fEt_K(\ket{i}\bra{i})\otimes \varphi_i\otimes \fTAB(K,R)\otimes \ket{0}\bra{0})\fEt_K\ket{j})                                                       \\
		\leq & \fpoly(q,N,L)2^{-0.5\kappa}
	\end{align}
	
	Substitute it back into (\ref{eq:48}), we will know
	\begin{align}
		     & |\Pr(\cD(\bE_K\bE_R(\rho))=1)-\Pr(\cD(\bE_K\bE_R(\cE_{proj}(\rho)))=1)| \\
		\leq & 2^{N_q}\fpoly(q,N,L)2^{-0.25\kappa}
	\end{align}
	The second term in (\ref{eq:32}) can be bounded by Proposition \ref{prop:clsec}. $\cE_{proj}(\rho)$ is a classical state so we have
	\begin{align*}
&|\Pr(\cD(\bE_K\fGBC.\fEn_K(\cE_{proj}(\varphi)))=1)-\Pr(\cD(\bE_K\fGBC.\fEn_K(0^{N}))=1)|\\
\leq & \fpoly(q,N,L)2^{-\kappa}
	\end{align*}
	Combining these two inequalities we have
	$$\fAdv^{qIND-CPA-oneshot}_\fGBC(\sA,\kappa)\leq\fpoly(q,N,L)2^{-0.25(\kappa-4N_q)}$$
\end{proof}
\subsection{Standard Model}
In the last section we prove the security in the quantum random oracle model. In practice, the random oracle can usually be replaced with hash functions, and we claim that our protocol is not an exception (Conjecture \ref{conj:qrotohash}). %The proof of Theorem \ref{thm:qssec} and Proposition \ref{prop:clsec} relies on Theorem \ref{thm:rosec}, and 
In our protocol, it's more natural to use a symmetric key encryption scheme directly: the usage of the random oracle in our protocol is on the symmetric multi-key encryption scheme with key tags, and the key verification can be replaced with the ``point-and-permute'' technique from the classical garbled circuit.\par
When using symmetric key encryption instead of the random oracle, since in our protocol we use affine functions in KDM game, we need at least that the symmetric key encryption is secure against quantum adversaries under KDM game for affine functions. Although this is a strong assumption, it's still reasonable in practice.\par
\section{Applications}
\subsection{Blind Quantum Computation for C+P Circuits}\label{sec:BQC}
Protocol \ref{prtl:gbc} is a quantum computation delegation protocol. But since the circuit can be put into inputs, we can turn it into a blind quantum computation protocol, where the server doesn't know either input state or  the circuit to be applied. If we only want to hide the type of gates in the circuit, our original protocol actually already achieves it. But if we also want to hide the circuit topology, we need to do more. The adversaries should only know the fact that the circuit is a C+P circuit, the input size and an upper bound on the circuit size. In this subsection we are going to construct a universal machine $\cU$ such that for all the C+P circuit $C$, $C(\rho)=\cU(C,\rho)$. What's more, we want $\cU$ to be in C+P so that we can use our protocol on $\cU$.\par
Suppose the size of input is $N$ and the phase gates are all in the form of $R_Z(\pi/2^d), d\in [D]$. Then there are $N^3+ND$ possible choices for each gate. Thus a $\log(N^3+ND)$ bits description is enough for each gate. For the server-side evaluation, a bad implementation  may lead to $N^3+ND$ extra cost, and we can do a simple preprocessing on the circuit to reduce it: We can first introduce three auxiliary wires, and convert $C$ to a form that only contains three types of gates: (1)$R_Z(\pi/2^d)$ (2)a $\fSWAP$ operation between a normal wire and an auxiliary wire (3)a Toffoli gate on the auxiliary wires. After this transformation, the number of choices of the gates is only  $3N+1+ND$. Thus we can describe each gate by a string of length $\log(3N+1+ND)$. And given the description of $g$, the operation of $\cU$ is a series of multi-controlled gate operations, where the control wires correspond to the gate description and the target wires are the wires in the original circuit. And this multi-controlled multi-target operation is also in C+P and it can be transformed to the standard form of $\fToffoli$ and phase gates.\par %: $N$ controlled $\fSWAP$ operations conditioned on the first $N$ bits of the input, controlled $\fToffoli$ operation conditioned on the next $1$ bit, and controlled $R_Z(\pi/2^d)$ operations on wire $n$, conditioned on bit $(n,d)\in [N]\times [D]$ in the next $ND$ bits.\par
Since $\cU$ itself is a C+P circuit, we can delegate it by applying Protocol \ref{prtl:gbc}. Then the original circuit will be indistinguishable from the identity circuit, which means we know nothing beyond some information on its size.
%\subsection{Adding Interactions}
%Our protocol naturally implies an interactive protocol for delegation of universal quantum circuits whose complexity depends on the $\fH$ depth. This is different from several previous protocols, for example, \cite{UBQC}, where the number of interactions is linear to the $\fT$ depth. \cite{FH} discussed the $\fH$ depth of quantum circuits, and introduced the concept of ``Fourier Hierarchy'' (FH). So we get
%\begin{theorem}
%For circuits in $FH_{k+1}$, there exists a quantum computation delegation protocol such that it's secure in the quantum random oracle model, the client side quantum resources are $O(k\kappa N_q)$ $\fCN$ gates(assuming Conjecture \ref{conj:qindcpa}, $\kappa=\eta$, or under current security proof, $\kappa=\eta+4N_q$), and $O(k)$ rounds of interactions.
%\end{theorem}
%Every time the circuits meet an $\fH$ layer, the server can send the whole state back, then the client decrypts, applies $\fH$ and encrypts again, and sends it to the server and the protocol continues. The number of interactions is linear to the $\fH$ depth. \par
%Usually, the security proof for an interactive protocol can be harder than a non-interactive protocol. But in this protocol, the client can get the whole computation result and re-encrypt it with fresh new keys, so its security can be based on the security of the single round protocol in the honest-but-curious setting.
 \subsection{Delegation of Shor's Algorithm}
Shor's algorithm contains two parts: first we apply lots of Toffoli gates on $\ket{+}^{\otimes n}\otimes \ket{M}$, where $M$ is, for example, the number to be factored, and $n=\log M$; then measure, apply quantum Fourier transform and measure again. From \cite{ParallelQFT}\cite{ECCResource} we know the quantum Fourier transform is actually easy to implement: a quantum Fourier transform on $n$ qubits has time complexity $\tilde O(n)$. The main burden of Shor's algorithm is the Toffoli part. (\cite{ECCResource} contains resource estimates on the elliptic curve version.) With this protocol we can let the server do the Toffoli part of Shor's algorithm without revealing the actual value of the input.\par

Explicitly, suppose the client wants to run Shor's algorithm on $M$ while also wants to keep $M$ secret, the client can use the following protocol:
\begin{prtl}\label{prtl:shordeleg}Protocol for delegation of Shor's algorithm:\par
Suppose ShorToff is the Toffoli gate part of Shor's algorithm, and its length is $L$.
\begin{enumerate}
	\item The client samples $K\leftarrow \fGBC.\fKg(1^\kappa, 1^{2n}, 1^L)$. Then the client prepares $(\rho,tab)\leftarrow\fGBC.\fEn_K^{\text{ShorToff}}(\ket{+}^{\otimes n}\otimes \ket{M})$ and sends it to the server.
	\item The server evaluates $\fGBC_{\fCL}.\fEv^{\text{ShorToff}}(\rho,tab)$ and sends it back to the client.
	\item The client decrypts with $\fGBC.\fDc_K$. Then it does quantum Fourier transform itself and measures to get the final result.
\end{enumerate}
\end{prtl}
So the quantum resources on the client side are only $O(\kappa n)$ $\fCN$ gates plus $\tilde O(n)$ gates for quantum Fourier transform, and it can delegate Shor's algorithm to the server side securely.\par
\begin{theorem}
Protocol \ref{prtl:shordeleg} can be used to delegate Shor's algorithm securely and non-interactively, in the quantum random oracle model(without assuming trapdoor one-way functions), and for $n$ bit inputs, the amount of quantum resources on the client side are quasi-linear quantum gates plus $O(\kappa n)$ $\fCN$ gates (assuming Conjecture \ref{conj:qindcpa}, $\kappa=\eta$, or under the current security proof, $\kappa=\eta+4n$).
\end{theorem}
For comparison, if the client runs Shor's algorithm locally, the client needs to perform $\omega(n^2\log n)$ Toffoli gates, and the exact form depends on the multiplication method it uses. Schoolbook multiplication leads to $O(n^3)$ complexity; if it uses fast multiplication method, the complexity is still $\omega(n^2\log n)$ and it has a big hidden constant. %. For the elliptic curve version of Shor's algorithm, \cite{ECCResource} shows the Toffoli gate needed is $\tilde \Omega(n^3)$. And for the integer factorization version of Shor's algorithm, \cite{Shorsalg} shows
\par %For comparison, in the proved version of our protocol, the client needs $O(\eta n+n^2)$ CNOT gates. While if we assume Conjecture \ref{conj:qindcpa}, the client side quantum resources are reduced to $O(\eta n)$ CNOT gates.\par
\section{Quantum KDM Security}
As a natural generalization of our discussion of KDM-security, we formalize the quantum KDM security and construct a protocol in this section. Previously when we discuss the KDM security the function $f$ and message $m$ are classical; here we further generalize them to include quantum states and operations.
\begin{definition}
	A symmetric key non-adaptive quantum KDM game naSymQKDM for function family $\cF$  in the quantum random oracle model is defined as follows:
	\begin{enumerate}
		\item The challenger chooses bit $b\leftarrow_r\{0,1\}$ and samples $K=\{sk_i\}_{i=1}^{N}$, $sk_i\leftarrow\fKg(1^\kappa)$.
		\item The adversary and the challenger repeat the following for $L$ times,non-adaptively, in other words, the challenger should only sends out the answers in step (b) after it receives all the queries:
		      \begin{enumerate}
			      \item The adversary picks index $i$, function $f\in \cF$ and message $\rho\in \bD(\cR\otimes \cM)$, and sends system $\cM$ to the challenger.
			      \item If $b=1$, the challenger returns $c=\fEn_{sk_i}(f(K,\rho_m))$ to the adversary. If $b=0$, the challenger returns $c=\fEn_{sk_i}(0^{|f(K,\rho_m)|})$.
		      \end{enumerate}
		\item The adversary tries to guess $b$ with some distinguisher $\cD$, and outputs $b^\prime$.
	\end{enumerate}
\end{definition}
Note that $\cF$ can be quantum operations and can query the random oracle with quantum states. The output of functions in $\cF$ should be fixed-lengthed, otherwise $|f(K,m)|$ will not be well-defined.\par
The guessing advantage is defined as $\fAdv^{naSymQKDM}(\sA,\kappa)=|\Pr(b^\prime=1|b=1)-\Pr(b^\prime=1|b=0)|$
\begin{definition}
	A symmetric key quantum encryption scheme is nonadaptively qKDM-CPA secure for function $\cF$ if for any BQP adversary $\sA$, $$\fAdv^{naSymQKDM}_{\cF}(\sA,\kappa)=\fneg(\kappa)$$
\end{definition}

\subsection{Protocol Design}
\begin{prtl}\label{QKDM}
A Quantum KDM Secure Protocol in the Quantum Random Oracle Model:
	\begin{description}
		\item[Key Generation] $\fQKDM.\fKg(1^\kappa)$: $sk\leftarrow \{0,1\}^\kappa$.
		\item[Encryption] $\fQKDM.\fEn_{sk}(\rho)$: Sample $a,b\in_r\{0,1\}^{N}$, where $N$ is the length of inputs.\\
		Output $(\fX^a\fZ^b(\rho), \fKDMP.\fEn_{sk}(a,b))$.
		\item[Decryption] $\fQKDM.\fDc_{sk}((\rho,c))$: First compute $a,b\leftarrow \fKDMP.\fDc_{sk}(c)$, then output $\fX^a\fZ^b(\rho)$
	\end{description}
\end{prtl}
\begin{theorem}\label{QKDMsec}
	Protocol \ref{QKDM} is nonadaptively qKDM-CPA secure for functions in $\cF[\fpoly]$ in the quantum random oracle model, where $\cF[\fpoly]$ is the function family that makes at most $\fpoly(\kappa)$ queries to the quantum random oracle.
\end{theorem}
We put its proof in the full version of this paper.
\section{Open Problems}
One obvious open problem in our paper is to prove Conjecture \ref{conj:qindcpa}, the qIND-CPA security without additional requirement on $\kappa$. We believe this is true, but we can only prove the security when $\kappa-4N_q=\eta$. And another further research direction is to base these protocols directly on the assumptions in the standard model, for example, the existence of hash functions or symmetric key encryption schemes that are exponentially KDM secure for affine functions against a quantum adversary. We can also study how to optimize this protocol, and how efficient it is compared to other protocols based on the quantum one-time pad. One obvious route is to make use of the optimization techniques for classical garbled circuits.\par
Another open question is whether this protocol is useful in other problems than Shor's algorithm. Lots of previous works studied quantum circuits on $\{\mathsf{Clifford}, \fT\}$ gate set, and our work shows $\{\text{C+P},\fH\}$ is also important and worth studying. There are not many works on converting quantum circuits into layers of C+P gates and $\fH$ gates, and it's possible that some famous quantum algorithms which require a lot of $\fT$ gates, after converted into $\{\text{C+P},\fH\}$ gate set, can have small $\fH$ depth. This problem is still quite open, and further research is needed here.\par
What's more, KDM security in quantum settings is an interesting problem. This paper gives some initial study on it, but there are still a lot of open questions. Is it possible to construct quantum KDM secure protocol in the standard model? Could quantum cryptography help us design classical KDM secure scheme?Again, further research is needed here.\par
This paper also gives some new ideas on constructing secure quantum encryption schemes without using trapdoor functions. Although there is some result\cite{ITQHELimit} on the limit of information-theoretically secure quantum homomorphic encryption, in our work we use the quantum random oracle and make the circuits available to the client, the limit doesn't hold any more. So here comes lots of interesting problems on the possibility and impossibility of quantum computation delegation: What is the limit for non-interactive information-theoretically secure delegation of quantum computation, where the circuit is public/private, with/without quantum ROM? If we allow small amount of quantum/classical communication, does it lead to something different?
%\section{Conclusion}
%In this paper we study and construct several protocols in quantum random oracle model. The main protocol is the quantum computation delegation protocol for C+P circuits. The quantum resources needed in the client side is small, and the security can be proved in the quantum random oracle model, without the LWE assumption. What's more, we re-study the classical KDM secure protocol in quantum random oracle model, and generalize it to quantum KDM security and construct a protocol for it. And our work introduces some different techniques for the study of quantum cryptography, both in the construction of the protocol and the proof technique, and it's possible to find more applications for them in other problems. What's more, as far as we know, this is the first study of KDM security in quantum settings, and further research is needed in this direction.
%\par

\section*{Acknowledgements}
The author would like to thank Prof. Adam Smith, NSF funding and anonymous reviewers.

\bibliography{main_Jiayu}
\pagebreak
\begin{center}
\textbf{\large Supplementary Materials}
\end{center}
\appendix
\section{Missing Proofs}
\subsection{Missing Proofs in Section 3 to 5}
\begin{proof}[of Lemma \ref{lm:itecd}]
	Let's first consider the encryption of one qubit. Suppose the input is $\ket{\phi_{SR}}=\alpha\ket{0}\ket{\varphi_0}+\beta\ket{1}\ket{\varphi_1}$, and we apply $\fEt$ on the first register. The lemma holds since:
	\begin{align}
		  & \frac{1}{2^\kappa(2^\kappa-1)}\sum_{k_0,k_1\in\{0,1\}^\kappa,k_0\neq k_1}(\fEt^S_{k_0,k_1}\otimes \fI^R)(\phi)                                                                                                                                      \\
		= & \frac{1}{2^\kappa(2^\kappa-1)}\sum_{k_0,k_1\in\{0,1\}^\kappa,k_0\neq k_1}(\alpha\ket{k_0}\ket{\varphi_0}+\beta\ket{k_1}\ket{\varphi_1})(\alpha^\dagger\bra{k_0}\bra{\varphi_0}+\beta^\dagger\bra{k_1}\bra{\varphi_1})                               \\
		= & (\frac{1}{2^\kappa}\fI)\otimes(|\alpha|^2\varphi_0+|\beta|^2\varphi_1)\notag\\&\qquad+\frac{1}{2^\kappa-1}(\ket{\Phi}\bra{\Phi}-\frac{1}{2^\kappa}\fI)\otimes(\alpha^\dagger\beta\ket{\varphi_1}\bra{\varphi_0}+\alpha\beta^\dagger\ket{\varphi_0}\bra{\varphi_1})
	\end{align}
	where $\ket{\Phi}=(\frac{1}{\sqrt{2}}(\ket{0}+\ket{1}))^{\otimes \kappa}$.\par
	The first term can be written as $\frac{1}{2^\kappa}\fI\otimes\tr_S(\phi)$. Which means
	$$\Delta((\fEt^S_{k_0,k_1}\otimes \fI^R)(\phi),\frac{1}{2^\kappa}\fI\otimes\tr_S(\phi))\leq (\frac{1}{2})^{\kappa-4}$$
	Using this inequality on all the encrypted qubits completes the proof.
\end{proof}
%\begin{proof}[of Theorem \ref{thm:ctoqofpk}] The proof uses the same techniques as \cite{AnnesQHELowT}, and we make some revisions to fit our problem.\par
%	With an adversary for $\fQPK$, we can design an adversary for $\fCLPK$ as follows:\par
%	Suppose the adversary of $\fQPK$ can distinguish $\bE_K\fEt_K(\rho)\otimes \fCLPK.\fEn(K)$ and $\bE_K\fEt_K(0^N)\otimes \fCLPK.\fEn(K)$ with non-negligible advantage.($a,b,\rho$ can be chosen adaptively, but this doesn't influence our proof.) Then the adversary can make use of it in the IND-CPA game for $\fCLPK$, by choosing $K$ as the plaintext to be encrypted.\par
%	 Notice that when $b=1$ the simulation will work correctly. If $b=0$, the simulator will try to distinguish $\fEt_K(\rho)\otimes \fCLPK.\fEn(0^*)$ and $\fEt_K(0^N)\otimes \fCLPK.\fEn(0^*)$. But taking average on $K$, the distribution of $\fEt_K(\rho)\otimes \fCLPK.\fEn(0^*)$ is always the same as the distribution of $\fEt_K(0^N)\otimes \fCLPK.\fEn(0^*)$, so that they can't be distinguished. That means in the IND-CPA game of $\fCLPK$, we can distinguish the case $b=1$ from the case $b=0$ by running the distinguisher of $\fQPK$ and see whether it can distinguish them.
%\end{proof}
\begin{proof}[of Proposition \ref{fact:C+Pdecomp}]
	First, any diagonal gate can be decomposed to Toffoli gates and single qubit phase gates. For $R_Z(\theta)$, expand $\theta$ on $\pi/2^d,d=1,\cdots D$ will complete the proof.
\end{proof}
\subsection{Proof of Theorem \ref{thm:rosec}}
\begin{proof}
	The proof is similar to the proof of optimality of Grover's algorithm, which can be found in \cite{GroverOptimal}. But we need to make some revisions which are necessary for this problem.\par
	First, we notice that the key tags can also be seen as part of the ciphertexts, where the messages to be encrypted are $0^*$. So we can prove this theorem without considering the key tags, and we need to redefine $T$ as $2T$. This won't affect the final result since $2$ can be absorbed into the $\fpoly$ function.\par
	After getting the ciphertexts from the encryption phase, the adversary will use some distinguisher $\cD$ to distinguish the states and compute $b^\prime$. By expanding the space, we can assume the distinguisher $\cD$ is a measurement of a specific output qubit after applying a unitary transform $O$ on the ciphertexts and auxiliary qubits. $O$ can be written as $O=U_q\cH_qU_{q-1}\cdots \cH_1U_0$, where $\cH_i$ means ith query to the quantum random oracle.\par
	Furthermore, suppose the secret keys generated by the key generator are $K=\{sk_i\}_{i=1}^L$, and $R$ is the set of all the randomness in the encryption phase(including the random paddings and the random outputs from the random oracle). Let $c_{K,R}$ denote the ciphertexts that the adversary gets in the naSymKDM game when $b=1$, and $e_{K,R}$ denote the ciphertexts when $b=0$. Then what we need to prove can be rewritten as
	\begin{align}
		     & \fAdv_{\fKDMP,\cF[q^\prime]}^{naSymKDM}(\sA_{(L,T,q)},\kappa)                                                                                                            \\
		=    & \tr(P_0(\bE_{K}\bE_R(O(\ket{c_{K,R}}\bra{c_{K,R}}\otimes \ket{0}\bra{0})O^\dagger)\notag\\&\qquad\qquad\qquad\qquad\qquad\qquad-\bE_{K}\bE_R(O(\ket{e_{K,R}}\bra{e_{K,R}}\otimes \ket{0}\bra{0})O^\dagger)))          \\
		\leq & \frac{1}{2}|\bE_{K}\bE_R(O(\ket{c_{K,R}}\bra{c_{K,R}}\otimes \ket{0}\bra{0})O^\dagger)-\bE_{K}\bE_R(O(\ket{e_{K,R}}\bra{e_{K,R}}\otimes \ket{0}\bra{0})O^\dagger)|_{\tr}
	\end{align}
	Where $\ket{0}\bra{0}$ can be very large system.\par
	Suppose $\cH^\prime$ is a new random oracle, which is independently random of $\cH$ on inputs that contain a prefix in $K$, and behaves the same as $\cH$ otherwise. Notice that if we can replace the random oracle queries in $O$ with $\cH^\prime$, since the challenger uses $\cH(sk||\cdot)$ for encryption, the messages will be hidden and won't be distinguished.\par %the randomness that the encryption phase used will be independent of the output of $\cH^\prime$, and after taking expectation on the randomness of random oracle the message will be hidden.\par
	%We will also need to use the one-way to hiding lemma, which is introduced in \cite{onewaytohiding}. Which means, we only need to give a bound on $\bra{sk}U_q\cH_qU_{q-1}\cdots \cH_1U_0(\ket{0})\ket{sk}$.\par
	Define $O_i=U_q\cH_qU_{q-1}\cdots \cH_{i+1}U_i\cH^\prime_iU_{i-1}\cdots\cH^\prime_1U_0$. So $O=O_0$. Define
	\begin{align}
		\phi^{i}&=\bE_R(O_i (\ket{c_{K,R}}\bra{c_{K,R}}\otimes \ket{0}\bra{0})O_i^\dagger)\\
	\psi^{i}&=\bE_R(O_i (\ket{e_{K,R}}\bra{e_{K,R}}\otimes \ket{0}\bra{0})O_i^\dagger)\label{eq:80}
	\end{align}
	So
	\begin{align}
		     & \fAdv_{\fKDMP,\cF[q^\prime]}^{naSymKDM}(\sA_{(L,T,q)},\kappa)                                                    \\
		=    & \frac{1}{2}|\bE_K(\phi^0-\psi^0)|_{\tr}                                                                          \\
		\leq & \frac{1}{2}(|\bE_K(\phi^0-\phi^q)|_{\tr}+|\bE_K(\phi^q-\psi^q)|_{\tr}+|\bE_K(\psi^0-\psi^q)|_{\tr})\label{eq:74}
	\end{align}
	Let's first estimate $|\bE_K(\phi^q-\psi^q)|_{\tr}$. First by dividing the randomness $R$ into $R^\prime$, which is the randomness of $\cH^\prime$, and $R_{K||}$, which is the randomness of $\cH(sk||\cdot)$  together with the random paddings. We can see
	\begin{equation}
		|\bE_K(\phi^q-\psi^q)|_{\tr}= |\bE_K\bE_{R^\prime}O_i(\bE_{R_{K||}}(\ket{c_{K,R}}\bra{c_{K,R}}-\ket{e_{K,R}}\bra{e_{K,R}}))O_i^\dagger)|_{\tr}\label{eq:75}
	\end{equation}
	The components of $\ket{c_{K,R}}$ have form $(R_i,\cH(sk||R_i)\oplus m)$. We can't say the randomness in $R_{K||}$ hides everything in $m$ since the computation of $m=f(K,msg)$ may also contains some random oracle queries. But the random oracle queries of $f$ are all in classical states so the probability that $sk||R_i$ does not appear before can be bounded. Denote the queries to the random oracle inside the circuit of $f$ as $Q$. Define $bad$ as one of the following two events: (1) for some $i$, $sk_i||R_i\in Q$; (2) the set of $sk_i||R_i$ contains repetitions. And define $good$ as the complement of $bad$.% Denote the randomness of $R_i$ as $R_r$ and the randomness in other places as $R^\prime$(including the randomness inside the random oracle), then
	%\begin{align}
	%&|\bE_R(\ket{c_{K,R}}\bra{c_{K,R}})-\bE_R(\ket{e_{K,R}}\bra{e_{K,R}})|\\
	%	=&|\bE_{R_i}\bE_{R^\prime}(\ket{c_{K,R}}\bra{c_{K,R}})-\bE_{R_i}\bE_{R^\prime}(\ket{e_{K,R}}\bra{e_{K,R}})|\\
	%	=&|\Pr(sk||R_i\in Q)\bE_{R_i}\bE_{R^\prime}(\ket{c_{K,R}}\bra{c_{K,R}}|R_i\in Q)+\Pr(sk||R_i\notin Q)\bE_{R_i}\bE_{R^\prime}(\ket{e_{K,R}}\bra{e_{K,R}}|R_i\notin Q)\\
	%	&-\Pr(sk||R_i\in Q)\bE_{R_i}\bE_{R^\prime}(\ket{c_{K,R}}\bra{c_{K,R}}|R_i\in Q)-\Pr(sk||R_i\notin Q)\bE_{R_i}\bE_{R^\prime}(\ket{e_{K,R}}\bra{e_{K,R}}|R_i\notin Q)|\\
	%	=&|\Pr(sk||R_i\in Q)\bE_{R_i}\bE_{R^\prime}(\ket{c_{K,R}}\bra{c_{K,R}}|R_i\in Q)-\Pr(sk||R_i\in Q)\bE_{R_i}\bE_{R^\prime}(\ket{c_{K,R}}\bra{c_{K,R}}|R_i\in Q)|\\
	%	\leq & 2\Pr(sk||R_i\in Q)\\
	%	\leq &
	%\end{align}
	\begin{align}
		     & |\bE_{R_{K||}}(\ket{c_{K,R}}\bra{c_{K,R}}-\ket{e_{K,R}}\bra{e_{K,R}})|_{\tr}                           \label{eq:76}                \\
		=    & |\Pr(bad)\bE_{R_{K||}}(\ket{c_{K,R}}\bra{c_{K,R}}|bad)+\Pr(good)\bE_{R_{K||}}(\ket{c_{K,R}}\bra{c_{K,R}}|good)                      \\
		     & -\Pr(bad)\bE_{R_{K||}}(\ket{e_{K,R}}\bra{e_{K,R}}|bad)-\Pr(good)\bE_{R_{K||}}(\ket{e_{K,R}}\bra{e_{K,R}}|good)|_{\tr} \label{eq:78} \\
		=    & |\Pr(bad)\bE_{R_{K||}}(\ket{c_{K,R}}\bra{c_{K,R}}|bad)-\Pr(bad)\bE_{R_{K||}}(\ket{e_{K,R}}\bra{e_{K,R}}|bad)|_{\tr}   \label{eq:79} \\
		\leq & 2\Pr(bad)                                                                                                                           \\
		\leq & 4T^2({q^\prime}+2)2^{-\kappa}\label{eq:81}
	\end{align}
	(\ref{eq:78}) to (\ref{eq:79}) is because when $sk||R_i$ does not appear in any other places of the random oracle queries, we can take average on the randomness of $\cH(sk||R_i)$ and the distribution of $\cH(sk||R_i)\oplus 0$ will be the same as $\cH(sk||R_i)\oplus m$.\par
	Substitute it back into (\ref{eq:75}) we will know
	\begin{equation}
		|\bE_K(\phi^q-\psi^q)|_{\tr}\leq 4T^2({q^\prime}+2)2^{-\kappa}\label{eq:82}
	\end{equation}
	%But if we define $\ket{f_{K,R}}$ as the concatenation of $(R_i,\cH^{\prime\prime}(sk||R)\oplus m)$\par
	And we have:
	\begin{align}
		&|\phi^0-\phi^q|_{\tr} \\
		\leq &\sum_i|\phi^i-\phi^{i-1}|_{\tr}                                                                                                                                        \label{eq:90}\\
		                      \leq & \sum_i\bE_R|((\cH_i-\cH^\prime_i)V_i(\ket{c_{K,R}}\otimes \ket{0}))|\quad(V_i:=U_i\cH^\prime_{i-1}U_{i-1}\cdots\cH^\prime_1U_0)                                        \\
		                      \leq & \sum_i\bE_R2\sqrt{\tr(P_{K}V_i(\ket{c_{K,R}}\bra{c_{K,R}}\otimes \ket{0}\bra{0})V_i^\dagger)}                                                                          \\
		%&=2\sum_{i}\bE_R\sqrt{|P_{K}(\rho_i)|}\\
		                      \leq & 2\sum_i\sqrt{\bE_R\tr(P_K(V_i(\ket{c_{K,R}}\bra{c_{K,R}}\otimes \ket{0}\bra{0})V_i^\dagger)))}                                                                         \\
		                      \leq & 2\sqrt{q\sum_i\tr(P_{K}(\bE_{R^\prime}(V_i(\bE_{R_{K||}}(\ket{c_{K,R}}\bra{c_{K,R}})\otimes \ket{0}\bra{0})V_i^\dagger)))}                               \label{eq:94} \\
		                      \leq & 2\sqrt{q\sum_i(\tr(P_{K}(\bE_{R^\prime}(V_i(\bE_{R_{K||}}(\ket{e_{K,R}}\bra{e_{K,R}})\otimes \ket{0}\bra{0})V_i^\dagger)))+4T^2({q^\prime}+2)2^{-\kappa})}\label{eq:98}
	\end{align}
	Where the last step is by (\ref{eq:76})-(\ref{eq:81}).\par
	And when we take the expectation on $K$, we have
	\begin{align}
		     & \bE_K|\phi^0-\phi^q|_{\tr}                                                                                                              \\
		\leq & \sqrt{\bE_K|\phi^0-\phi^q|_{\tr}^2}                                                                                                     \\
		\leq & \sqrt{4q\bE_K\sum_i(\tr(P_{K}(\bE_{R^\prime}(V_i(\bE_{R_{K||}}(\ket{e_{K,R}}\bra{e_{K,R}})\otimes \ket{0}\bra{0})V_i^\dagger)))+2T({q^\prime}+T)2^{-\kappa})} \label{eq:99}\\
		\leq & \sqrt{(4q(qL+4T^2({q^\prime}+2)))2^{-\kappa}}\label{eq:102}
	\end{align}
(\ref{eq:99}) to (\ref{eq:102}) is because the inner part is the same for different $K$ so we can take average on $P_K$.\par
	Similarly we have
	\begin{equation}\bE_K|\psi^0-\psi^q|_{\tr}\leq \sqrt{(4q(qL+4T^2({q^\prime}+2)))2^{-\kappa}}\label{eq:93}
	\end{equation}
	Substitute (\ref{eq:82})(\ref{eq:102})(\ref{eq:93}) into (\ref{eq:74}), we have
	$$\fAdv_{\fKDMP,\cF[q^\prime]}^{naSymKDM}(\sA_{(L,T,q)},\kappa)\leq 2\sqrt{(4q(qL+4T^2({q^\prime}+2)))2^{-\kappa}}+4T^2({q^\prime}+2)2^{-\kappa}$$
\end{proof}
\subsection{Proof of Proposition \ref{prop:clsec}}
First, let's prove a lemma. To handle the case that there are some revealed keys in the $\fGBC$ game, we introduce the following game, which we called ``rG'' game, which is the starting point of our proof, and whose security can be proved from Theorem \ref{thm:rosec}.
\begin{lemma}\label{lm:rgsec}
	Let's first define an intermediate game for $\fCL$. Let's call restricted game rG, with parameters $(N,L,J,q), \kappa$:
	\begin{enumerate}
		\item The challenger chooses bit $b\leftarrow\{0,1\}$ and samples $(sk_i)_{i=1}^L$, $sk_i\leftarrow \fKg(1^\kappa)$
		\item The adversary picks (1)a set of pairs $P=\{(S_j,T_j)\}_{j=1}^J$, where $S_j,T_j$ are all tuples of indexes, $\forall j, S_j\subseteq [L], T_j\subseteq [L], |S_j|=|T_j|=3\text{ or }|S_j|=1,|T_j|=0$; (2)a sequence of messages $M=(m_1,m_2,\cdots m_J)$; (3)index set $Rev=\{rid_s\}_{s=1}^{N}$, $rid_s\in [L]$.\par
		      Define $Closure(Rev)$ as the minimum set that satisfies:\\(1)$Rev\subseteq Closure(Rev)$; (2)$\forall (S_j,T_j)\in P$, if $S_j\subseteq Closure(Rev)$, then $T_j\subseteq Closure(Rev)$.
		\item The challenger sends a set of secret keys $\{sk_{rid}\},rid\in Rev$ to the adversary. And for all $i\in [J]$, the challenger also sends:
		      \begin{gather}
			      \fCL.\fEn_{sk_{S_i}}(sk_{T_i}||m_i),\text{if }b=1\text{ or }S_i\subseteq Closure(Rev)\\
			      \fCL.\fEn_{sk_{S_i}}(0^{|sk_{T_i}||m_i|}),\text{ otherwise}
		      \end{gather}

		      Here we use $sk_{S_i}$ as the abbreviation of the three keys used for the encryption whose indexes are in $S_i$, and use $sk_{T_i}$ to denote the concatenation of keys whose indexes are in $T_i$.

		\item The adversary tries to guess $b$ with some distinguisher $\cD$, which will query the random oracle(either in classical or quantum states) $q$ times.
	\end{enumerate}
	The guess advantage is defined as $\fAdv^{rG}_\fCL(\sA_{(N,L,J,q)},\kappa)=|\Pr(b^\prime=1|b=1)-\Pr(b^\prime=1|b=0)|$.\par
	We have $\fAdv^{rG}_\fCL(\sA_{(N,L,J,q)},\kappa)=\fpoly(q,L,J)2^{-0.5\kappa}$, where $\fpoly$ is a polynomial that does not depend on $\sA$ or the parameters.
\end{lemma}
The proof is based on Theorem \ref{thm:rosec} and is postponed to the next subsection of the appendix.\par

Now we can prove Proposition \ref{prop:clsec} from Lemma \ref{lm:rgsec}.
\begin{proof}[of Proposition \ref{prop:clsec}]
	Note that, what the adversary gets in this game is already very similar to the ciphertexts in Protocol \ref{prtl:gbc} when the input is classical. So with an adversary for IND-CPA game for $\fGBC^C_\fCL$ with parameters $(N,L,q), \kappa$, which will choose $\ket{i}\bra{i}$ as input and use $\cD$ as the distinguisher, we can design an adversary for $rG$ with parameters $(N,3L,16L,q), \kappa$ in this way:
	%View $\fEt_K(\ket{i})$ in $\fGBC$ protocol as the revealed keys, and $\fTAB^C_\fCL(K)$ can be computed from the ciphertexts in the $rG$ game(after choosing the right $P$). 
	\begin{enumerate}
		\item A key set $K=(k_b^l)$ is sampled. Let's reindex the keys with $[3L]$. Choose $Rev_i$ to be the set of indexes of keys revealed in $\fEn_K(i)$, where $i$ is the input string to the IND-CPA game of $\fGBC$.
		\item Choose $P_C$ and $M_C$ corresponding to circuit $C$. Explicitly, each element of $P_{C}$ corresponds to a row in garbled tables, and for each row in the form of $\fCL.\fEn_{k_1,k_2,k_3}(sk)$, it corresponds to an index $i$ such that $S_i$ contains the indexes of ${k_1,k_2,k_3}$, and $T_i$ is the indexes of $sk$, $m_i$ is empty. And for each row in the form of $\fCL.\fEn_{k}(m)$, $S_i$ contains the index of $k$, and $m$ is $m_i\in M_C$, $T_i=\emptyset$.	\end{enumerate}
	and define $\fShuffle$ as the operation that arranges these ciphertexts into the corresponding places in garbled tables and shuffles each table randomly. Explicitly, let's use $\fOracle_b^{rG}(P_C,M_C)$ as the output of the $rG$ game in the third step, taking average on all the randomness. Then we have
	\begin{align*}&\fTAB_{\fCL}^C(K)=\bE_M\fShuffle(\fOracle^{rG}_{b=1}(P_{C},M_{C}))\\
	&\fGBC_\fCL^C.\fEn_K(i)=K_{Rev_i}\otimes \fTAB_{\fCL}^C(K)\\
	&\fGBC_\fCL^C.\fEn_K(0^N)=K_{Rev_0}\otimes \fTAB_{\fCL}^C(K)\\\end{align*}
	This means if we have an adversary which can distinguish $\bE_K\fGBC.\fEn_K(i)$ and $\bE_K\fGBC.\fEn_K(0^N)$ by some distinguisher $\cD$, we can design an adversary  for $rG$ by first choosing the corresponding $P,Rev,M$ and get $K_{Rev}\otimes \fOracle^{rG}_{b=1}(P,M)$, apply $\fShuffle$ and then use $\cD$ to distinguish the two cases. Explicitly, by Lemma \ref{lm:rgsec}, when the input is $i$:
	\begin{align}
		     & |\Pr(\cD(\bE_K(K_{Rev_i}\otimes \fShuffle(\fOracle^{rG}_{b=0}(P_{C},M_{C}))=1)))\notag\\
		     & -\Pr(\cD(\bE_K(K_{Rev_i}\otimes \fShuffle(\fOracle^{rG}_{b=1}(P_{C},M_{C}))))=1)| \\
		=    & \fAdv^{rG}_{\fCL}(\sA^\prime_{(N,6L,16L,q)}, \kappa)                              \\
		\leq & \fpoly(q,N,L)2^{-0.5\kappa}\label{eq:11}
	\end{align}
	When the input is $0$:
	\begin{align}
		     & |\Pr(\cD(\bE_K(K_{Rev_0}\otimes \fShuffle(\fOracle^{rG}_{b=0}(P_{C},M_{C}))=1)))  \notag\\
		     & -\Pr(\cD(\bE_K(K_{Rev_0}\otimes \fShuffle(\fOracle^{rG}_{b=1}(P_{C},M_{C}))=1)))| \\
		=    & \fAdv^{rG}_{\fCL}(\sA^{\prime\prime}_{(N,6L,16L,q)}, \kappa)                  \\
		\leq & \fpoly(q,N,L)2^{-0.5\kappa}\label{eq:14}
	\end{align}
	And after shuffling the tables and taking average on all the possible keys we have
	\begin{align}
		  & |\bE_K\bE_M(K_{Rev_i}\otimes \fShuffle(\fOracle^{rG}_{b=0}(P_{C},M_{C})))\notag\\
		  &-\bE_K\bE_M(K_{Rev_0}\otimes \fShuffle(\fOracle^{rG}_{b=0}(P_{C},M_{C})))|_{tr}=0\label{eq:21}
		%		\leq &|\cD^\prime(Oracle^{rG}_{b=0}(P_{i,C},R_{i,C},M_{i,C})=0)-\cD^\prime(Oracle^{rG}_{b=1}(P_{0,C},R_{0,C},M_{0,C})=0)|\\
	\end{align}
	This is because the keys and the $m$ for phase gates are all chosen randomly. Notice that we need to take average on the choices of $M_C$ here. \par
	And this implies
	\begin{align}
		     & \fAdv_{\fGBC_\fCL}^{IND-CPA}(\sA, \kappa)                                       \\
		=    & |\Pr(\cD(\bE_K\fGBC.\fEn_K(i))=1))-\Pr(\cD(\bE_K\fGBC.\fEn_K(0^{N}))=1)|        \\
		=    & |\Pr(\cD(\bE_K(K_{Rev_i}\otimes \bE_M\fShuffle(\fOracle^{rG}_{b=1}(P_{C},M_{C}))=1)))\notag\\
		     & -\Pr(\cD(\bE_K(K_{Rev_0}\otimes \bE_M\fShuffle(\fOracle^{rG}_{b=1}(P_{C},M_{C}))))=1)| \\
		     \leq&|\bE_K\bE_M(K_{Rev_i}\otimes \fShuffle(\fOracle^{rG}_{b=0}(P_{C},M_{C})))\notag\\
		     &-\bE_K\bE_M(K_{Rev_0}\otimes \fShuffle(\fOracle^{rG}_{b=0}(P_{C},M_{C})))|_{tr}\notag\\
		     &+\bE_M|\Pr(\cD(\bE_K(K_{Rev_0}\otimes \fShuffle(\fOracle^{rG}_{b=0}(P_{C},M_{C}))))=1)-\notag\\
		     &\quad\quad\Pr(\cD(\bE_K(K_{Rev_0}\otimes \fShuffle(\fOracle^{rG}_{b=1}(P_{C},M_{C}))))=1)|\notag\\
		     &+\bE_M|\Pr(\cD(\bE_K(K_{Rev_i}\otimes \fShuffle(\fOracle^{rG}_{b=0}(P_{C},M_{C}))))=1)-\notag\\
		     &\quad\quad\Pr(\cD(\bE_K(K_{Rev_i}\otimes \fShuffle(\fOracle^{rG}_{b=1}(P_{C},M_{C}))))=1)|\\
		\leq & \fpoly(N,L,q)2^{-0.5\kappa}
	\end{align}
	Where the last step is by (\ref{eq:11})(\ref{eq:14})(\ref{eq:21}).
\end{proof}

\subsection{Proof of Lemma \ref{lm:rgsec}}
\begin{proof}
	First, let's give an estimate for $\fAdv^{rG}_\fCL(\sA,\kappa)$. Notice that the construction of multi-key encryption $\fCL$ is the same as cascading encryption \\$\fEn_{k_1}(\fEn_{k_2}(\fEn_{k_3}(m)))$, where $\fEn_k(m)=(R,\cH(k||R)\oplus m)$. And we can view the inner $\fEn$ as part of KDM function $f\in \cF[q^\prime= 2]$. What's more, the key tags can be obtained by encrypting $0^*$ with the corresponding keys. So an adversary $\sA$ for $rG$ game with parameters $\kappa, (N,L,J,q)$ can be used to design an adversary $\sA^\prime$ for naSymKDM game for function family $\cF[q^\prime= 2]$ with parameters $\kappa, (L,J+L,q)$ as follows:
	\begin{enumerate}
		\item In the step 1 of naSymKDM game, $K=\{sk_i\}_{i=1}^{L}$ is sampled. The adversary simulates an $rG$ game, and samples $K^\prime=\{sk^\prime_i\}_{i=1}^{L}$.\\
		      The step 2(a) of naSymKDM game is done by the adversary. The adversary first runs the step 2 of $rG$ game. Suppose what the adversary picks are $P,R,M$. Denote $|Closure(Rev)|=L_1$, $L-|Closure(Rev)|=L_2$. The adversary in $rG$ game will regard $\{sk_i\}_{i=1}^{L_2}\cup \{sk_i^\prime\}_{i=1}^{L_1}$ as the keys sampled in the first step in $rG$ game by the challenger.
		\item Then the adversary in $rG$ goes to step 3. It will get a set of ciphertexts from the challenger in $rG$. Notice that what the adversary gets in $rG$ can be simulated with naSymKDM game as follows:
		\begin{enumerate}
		\item For $i\in [J]$, suppose $S_i=(a,b,c)$. If $a\notin Closure(Rev)$, query $\fKDMP.\fEn$ on  $f(K,m)=\fEn_{sk_b}(\fEn_{sk_c}(sk_{T_i}||m_i))$ under $sk_a$. If $k_1\in Closure(Rev)$ but $k_2\notin Closure(Rev)$, query on  $f=\fEn_{sk_c}(sk_{T_i}||m_i)$ under $sk_b$. If $k_1,k_2\in Closure(Rev)$ but $k_3\notin Closure(Rev)$, query on $sk_{T_i}||m$ under $sk_c$. If all the three indexes are in $Closure(Rev)$, skip this query.\\
		      Then for each key tag in the form of $(R_i, \cH(sk_i||R_i))$, query the encryption of $0$ under $sk_i$.
		\item In step 2(b) of naSymKDM game, $\sA^\prime$ gets a list of ciphertexts. For $i\in [|J|]$, suppose $S_i=(a,b,c)$. Suppose $c$ is the ciphertext corresponding to $S_i$ that the adversary gets in step 2(if there is). If $a\notin Closure(Rev)$, assign $c^\prime=c$. If $k_1\in Closure(Rev)$ but $k_2\notin Closure(Rev)$, assign $c^\prime=\fEn_{sk^\prime_a}(c)$. If $k_1,k_2\in Closure(Rev)$ but $k_3\notin Closure(Rev)$, assign $c^\prime=\fEn_{sk^\prime_a}(\fEn_{sk^\prime_b}(c))$. If all the three indexes are in $Closure(R)$, assign $c^\prime=\fEn_{sk^\prime_a}(\fEn_{sk^\prime_b}(\fEn_{sk_c}(sk_{T_i}||m_i)))$. Store this $c^\prime$ into a list.
		\end{enumerate}
		\item Use the distinguisher of $rG$ game on the list of $c^\prime$ that the adversary gets in the last step. Suppose the result is $b^\prime$. Output $b^\prime$ as the distinguishing output in naSymKDM game.
	\end{enumerate}
	So from the security of KDM protocol we know
	$$\fAdv^{rG}_\fCL(\sA_{(N,L,J,q)},\kappa)\leq \fAdv^{naSymKDM}_{\fKDMP,\cF[q^\prime=2]}(\sA_{(L,J+N,q)},\kappa)\leq \fpoly(q,L,J)2^{-0.5\kappa}$$
\end{proof}
\subsection{Proof of Lemma \ref{lem:mcluncp}}
\begin{proof}
	For any operation $\cD$, $i\neq j,\ket{\varphi_j}$, $$p=\bE_{K}\bE_R\tr((\fEt_K\ket{j})^\dagger\cD(\fEt_K(\ket{i}\bra{i})\otimes \varphi_i\otimes \fTAB(K,R))(\fEt_K\ket{j})$$ is the probability that the input is $\fEt_K\ket{i}$ and the output is $\fEt_K\ket{j}$. But since $\fEt_K\ket{i}$ is a classical state we can make one copy of it and in the end we will know both from $\cD$ with this probability.\par
	Let's use it to design an adversary $\sA$ for the rG game of $\fCL$:
	\begin{enumerate}
		\item The adversary picks $\ket{i}\otimes \ket{\varphi_j}$ as the input and $P_C,M_C$ defined as the proof of Proposition \ref{prop:clsec}. Then the adversary will get $\fEt_K\ket{i}\otimes \ket{\varphi_i}\otimes \fOracle^{rG}_{b}(P_{C},M_{C})$ from the challenger in the $rG$ game.
		\item Then the adversary applies the distinguisher $\cD^\prime$ defined as follows:
		      \begin{enumerate}
			      \item The adversary makes a copy of the outputs and applies $\fShuffle^C$ to get $\fTAB^C_{\fCL}(K)=\fShuffle^C(\fOracle^{rG}_{b}(P_{C},M_{C}))$. Then it gets
			            $$\fEt_K\ket{i}\otimes \fOracle^{rG}_{b}(P_{C},M_{C})\otimes\fEt_K\ket{i}\otimes \fTAB^C_{\fCL}(K)\otimes \ket{\varphi_i}$$
			            Then the adversary applies $\cD$ on the last three systems and measures in the computational basis.
			      \item Find a $w$ such that $i$ and $j$ differ in bit $w$. There are two keys $k_0^w$ and $k_1^w$ on this wire, from $\fEt_K\ket{i}$ the adversary knows one key on wire $w$, and from the last step the adversary gets some result which might be $\fEt_K\ket{j}$(with some probability), from which it can get another key on $w$. Then the adversary verifies whether this value is another key with the key verification information in $\fOracle^{rG}_{b}(P_{C},M_{C})$. If the verification passes, output 1. Otherwise, output 0.
		      \end{enumerate}
		      %\item Suppose $w$ is connected to Toffoli gate $g$, and the corresponding garbled gate in $\fTAB(K)$ is $tab_g$. From $\fEt_K(\ket{i})$ the adversary can also know at least one keys for the other two wires of $g$. So the adversary knows two combinations of keys for $tab_g$ and it can decrypt two lines in this table. But Then we will be able to distinguish the $b=0$ and $b=1$ case in the rG game. Explicitly, 
	\end{enumerate}
	By the $rG$ game security we have
	\begin{align}
		     & \fAdv^{rG}_{\fGBC_\fCL}(\sA_{(N,6L,16L,q+1)},\kappa)                                                                                                                                                \\
		=    & |\Pr(\cD^\prime(\fEt_K\ket{i}\otimes \ket{\varphi_i}\otimes \fOracle^{rG}_{b=1}(P_{C},M_{C}))=1)\notag\\
		&-\Pr(\cD^\prime(\fEt_K\ket{i}\otimes \ket{\varphi_i}\otimes \fOracle^{rG}_{b=0}(P_{C},M_{C}))=1)|\label{eq:28} \\
		\leq & \fpoly(N,L,q)2^{-0.5\kappa}\label{eq:29}
	\end{align}
	And $\Pr(\cD^\prime(\fEt_K\ket{i}\otimes \ket{\varphi_i}\otimes \fOracle^{rG}_{b=0}(P_{C},M_{C}))=1)$ is the probability of ``compute and verify successfully'', it can be bounded by the optimality of Grover search. Let's analyze with more details.\par
Suppose in $\cD^\prime$ we choose to compute and verify the key $sk^w_x$ on input wire $w$. That will imply the index of $sk^w_{1-x}$ is in $Rev$, and the index of $sk^w_x$ is not in $Closure(Rev)$. Recall that in the construction of $\fOracle^{rG}_{b=0}(P_{C},M_{C})$, if the index of some key is not in $Closure(Rev)$, and when $b=0$, the $\fOracle$ will replace the ciphertext with encryption of $0^*$. So the only information related to $sk_x^w$ is contained in: (1)the oracle output in the form of $\fCL.\fEn_{S_i}(0^*), sk_b^w\in S_i$. (2) the key verification tag in the form of $\fCL.\fEn_{sk_b^w}(0^*)$.\par
Write $\cD^\prime=\fVer\circ \mathsf{Compute}$, where $\mathsf{Compute}$ is the operation in the 2(a) of $\cD$. Since the key tag can be very long(the output length of $\cH$ is arbitrary), the probability that $sk^w_x$ is not the only value that can pass the verification is exponentially small. So conditioned on it is the only entry that can pass the verification, we have
\begin{align}
	&\Pr(\cD^\prime(\fEt_K\ket{i}\otimes \ket{\varphi_i}\otimes \fOracle^{rG}_{b=0}(P_{C},M_{C}))=1)\\
	=&\Pr(\mathsf{Compute}(\fEt_K\ket{i}\otimes \ket{\varphi_i}\otimes \fOracle^{rG}_{b=0}(P_{C},M_{C}))=sk_b^w)
\end{align}

which is the probability of computing the key $sk^w_{x}$ in $q$ queries. From the result in \cite{GroverOptimal} we know
	$$\Pr(\cD^\prime(\fShuffle(\fOracle^{rG}_{b=0}(P_{C},M_{C})))=1)=\Theta(q^22^{-\kappa})$$
	%		$$|\Pr(\fEt_K(\ket{i})\otimes\fEt_K\ket{j}\otimes Shuffle(Oracle^{rG}_{b=0}(P_{C},M_{C}))-Shuffle(Oracle^{rG}_{b}(P_{C},M_{C}))|=1$$
	Combining it with (\ref{eq:29}) we can get a bound on the first term in (\ref{eq:28}):
	$$\Pr(\cD^\prime(\fEt_K\ket{i}\otimes\ket{\varphi_i}\otimes \fShuffle(\fOracle^{rG}_{b=1}(P_{C},M_{C})))=1)\leq \fpoly(N,L,q)2^{-0.5\kappa}$$
	Return to the original problem.\\
	 $\bE_{K}\bE_R\tr((\fEt_K\ket{j})^\dagger\cD(\fEt_K(\ket{i}\bra{i})\otimes \varphi_i\otimes \fTAB(K,R))(\fEt_K\ket{j}))$ is the probability that the input is $\fEt_K(\ket{i})$ and the output is $\fEt_K(\ket{j})$, and when we construct $\cD^\prime$ from $\cD$, conditioned on this event, the result will pass the verification in 2(b) step with probability 1. That implies
	\begin{align}
		     & \bE_{K}\bE_R\tr((\fEt_K\ket{j})^\dagger\cD(\fEt_K(\ket{i}\bra{i})\otimes \varphi_i\otimes \fTAB(K,R))(\fEt_K\ket{j})) \\
		\leq & \Pr(\cD^\prime(\fEt_K\ket{i}\otimes\ket{\varphi_i}\otimes \fShuffle(\fOracle^{rG}_{b=1}(P_{C},M_{C})))=1)             \\
		\leq & \fpoly(N,L,q)2^{-0.5\kappa}
	\end{align}
\end{proof}
\subsection{Proof of Theorem \ref{QKDMsec}}
\begin{proof}
	Explicitly, define $\fAdv^{naSymQKDM}_{\cF[q^\prime]}(\sA_{(L,T,q)},\kappa)$ as the advantage in the game where the adversary is allowed to sample $L$ keys, make $T$ queries to the encryption oracle, and the distinguisher is allowed to query the random oracle $q$ times, and the function family is $\cF[q^\prime]$.\par
	The proof is similar to the proof of Theorem \ref{thm:rosec}. Let $c_{K,R}$ denote the ciphertexts that the adversary gets in the naSymQKDM game when $b=1$, and $e_{K,R}$ denote the ciphertexts when $b=0$. These are the same as the proof of Theorem \ref{thm:rosec}. One difference here is that $c$ and $e$ are not necessarily classical any more, so we need to consider $c$ and $e$ as the combination of the states of the ciphertexts returned by the challenger together with the reference system kept by the adversary. We can bound the distinguishing advantage in naSymQKDM game following the same route as the proof of Theorem \ref{thm:rosec}, and we only need to give a new estimate for 
	\begin{equation}|\bE_{R_{K||}}(\ket{c_{K,R}}\bra{c_{K,R}}-\ket{e_{K,R}}\bra{e_{K,R}})|_{\tr}\label{eq:90r}\end{equation}
	Here $K$ is a fixed key set, and $R_{K||}$ stands for the randomness in the random paddings, output of $\cH$ on inputs with prefix in $K$, and the one time pad keys.\par
	The ciphertexts $c_{K,R}$ and $e_{K,R}$ should be viewed as the ciphertexts for all the queries. Since we only consider non-adaptive settings, the computation of the ciphertexts can be written as
	$$\bE_{R_{K||}}(\ket{c_{K,R}}\bra{c_{K,R}})=\bE_{R_{K||}}(\fI\otimes\fQKDM.\fEn^{\otimes T})(\rho)$$
	$$\bE_{R_{K||}}(\ket{e_{K,R}}\bra{e_{K,R}})=\bE_{R_{K||}}(\fI\otimes\fQKDM.\fEn^{\otimes T})(\rho_R\otimes 0^{|\cM|})$$
	Where
	$$\rho= (\fI\otimes O)(\sigma_{msg}\otimes K\otimes \ket{0}\bra{0})(\fI\otimes O^\dagger),\rho\in \bD(\cR\otimes \cM), \fI\text{ operates on } \cR$$
	where $O$ is a unitary operation which stands for the circuit of $f\in \cF$, $\sigma_{msg}$ is the density operator of the input to $f$ together with its reference system, and the system $\cM$ of $\rho$ contains the inputs to $\fEn$ in all the queries. Pay attention that this $O$ is different from those defined in the proof of Theorem \ref{thm:rosec}: this $O$ is the operation of $f$ while the $O$ in the previous proof represents the distinguisher. And we slightly abuse the notation: the encryption queries in $\fQKDM.\fEn^{\otimes T}$ actually don't use the same keys.\par
 Once we get a bound for (\ref{eq:90r}), we can substitute it into (\ref{eq:94})-(\ref{eq:98}) and then we will get a new bound for (\ref{eq:102}), and similarly substitute the new bound for (\ref{eq:82})(\ref{eq:102})(\ref{eq:93}) into (\ref{eq:74}) will give us the naSymQKDM advantage we need.\par
	$O$ can be written as $O=U_q\cH_qU_{q-1}\cdots \cH_1U_0$. Suppose $Q$ is the set of queries to the random oracle when applying $\fKDMP.\fEn$, and since we have already fixed the keys, the randomness comes from the random paddings. Suppose $R_Q$ is the randomness of $\cH$ on the queries in $Q$. Let $R_{ab}$ denote the randomness from the choices of one time pad keys. Then we have
	\begin{align}
		&|\bE_{R_{K||}}(\ket{c_{K,R}}\bra{c_{K,R}})-\bE_{R_{K||}}(\ket{e_{K,R}}\bra{e_{K,R}})|_{\tr}\\
		=&|\bE_Q\bE_{R_Q}\bE_{R_{ab}}(\ket{c_{K,R}}\bra{c_{K,R}})-\bE_Q\bE_{R_Q}\bE_{R_{ab}}(\ket{e_{K,R}}\bra{e_{K,R}})|_{\tr}\label{eq:58}
	\end{align}
	First consider a fixed $Q$. Let $\cH^{\prime\prime}$ be a quantum random oracle that is independently random from $\cH$ on entries in $Q$. Use $R^\prime$ to denote the randomness in $\cH^{\prime\prime}$. And let $$O_i=U_q\cH_qU_{q-1}\cdots \cH_{i+1}U_{i}\cH^{\prime\prime}_iU_{i-1}\cdots\cH^{\prime\prime}U_0$$
	Define
	$$\psi^i=\bE_{R^\prime}(\fI\otimes \fQKDM.\fEn^{\otimes T})((\fI\otimes O_i)(\sigma_{msg}\otimes K\otimes \ket{0}\bra{0})(\fI\otimes O_i^\dagger))$$
	%	$$\phi=\bE_{R^\prime}(\cO_0 (\ket{c_{K,R}}\bra{c_{K,R}}\otimes \ket{0}\bra{0})\cO_0^\dagger)$$
	%	$$\phi^{i}=\bE_{R^\prime}(\cO_i (\ket{c_{K,R}}\bra{c_{K,R}}\otimes \ket{0}\bra{0})\cO_i^\dagger)$$
	%	$$\psi=\bE_{R^\prime}(\cO_0 (\ket{e_{K,R}}\bra{e_{K,R}}\otimes \ket{0}\bra{0})\cO_0^\dagger)$$
	%	$$\psi^{i}=\bE_{R^\prime}(\cO_i (\ket{e_{K,R}}\bra{e_{K,R}}\otimes \ket{0}\bra{0})\cO_i^\dagger)$$
	%Here we consider the randomness in the Pauli keys as part of $R^\prime$. 
	Note that this is not the same $\psi$ as defined in (\ref{eq:80}). Then we have
	\begin{align}
		&|\bE_{R_{K||}}(\ket{c_{K,R}}\bra{c_{K,R}})-\bE_{R_{K||}}(\ket{e_{K,R}}\bra{e_{K,R}})|_{\tr}\\
		\leq &|\bE_Q\bE_{R_Q}\bE_{R_{ab}}(\psi^0)-\bE_Q\bE_{R_Q}\bE_{R_{ab}}(\psi^q)|_{\tr}\\
		&+\bE_Q|\bE_{R_Q}\bE_{R_{ab}}(\psi^q)-\bE_{R_Q}\bE_{R^\prime}\bE_{R_{ab}}(\ket{e_{K,R}}\bra{e_{K,R}})|_{\tr}\label{eq:64}
	\end{align}
And by the same technique as (\ref{eq:90})-(\ref{eq:94}) we have
	\begin{align}
		     & |\psi^0-\psi^q|_{\tr}                                                                                                                                                    \\
		\leq & \sum_i|\psi^i-\psi^{i-1}|_{\tr}                                                                                                                                          \\
		\leq & 2\sqrt{q^\prime\sum_i\tr(P_{Q}(\bE_{R^\prime}(V_i(\sigma_{msg}\otimes K\otimes \ket{0}\bra{0})V_i^\dagger)))}\quad(V_i:=U_i\cH^{\prime\prime}_{i-1}U_{i-1}\cdots\cH^{\prime\prime}_1U_0)
	\end{align}
	Which means
	\begin{align}
		     & |\bE_Q\bE_{R_Q}\bE_{R_{ab}}(\psi^0)-\bE_Q\bE_{R_Q}\bE_{R_{ab}}(\psi^q)|_{\tr} \\
		\leq & 2\bE_Q\bE_{R_Q}\sqrt{q^\prime\sum_i\tr(P_{Q}(\bE_{R^\prime}(V_i(\sigma_{msg}\otimes K\otimes \ket{0}\bra{0})V_i^\dagger)))}        \\
		\leq & 2\sqrt{q^\prime\sum_i\bE_Q\bE_{R_Q}\tr(P_{Q}(\bE_{R^\prime}(V_i(\sigma_{msg}\otimes K\otimes \ket{0}\bra{0})V_i^\dagger)))}        \\
		\leq & 2\sqrt{(q^\prime)^2T2^{-\kappa}}\label{eq:71}
	\end{align}
	And for the second term in (\ref{eq:64}), since the randomness of $R_Q$ is not correlated to the plaintexts anymore, taking average on $R_Q$ will hide the one time pad keys perfectly. Then taking average on the one time pad keys will hide $\rho$ perfectly:
	\begin{equation}
		\forall Q, |\bE_{R_Q}\bE_{R_{ab}}(\psi^q)-\bE_{R_Q}\bE_{R^\prime}\bE_{R_{ab}}(\ket{e_{K,R}}\bra{e_{K,R}})|_{\tr}=0\label{eq:65}
	\end{equation}

	Combining (\ref{eq:71})(\ref{eq:64})(\ref{eq:65}), we have
	$$|\bE_R(\ket{c_{K,R}}\bra{c_{K,R}})-\bE_R(\ket{e_{K,R}}\bra{e_{K,R}})|_{\tr}\leq\fpoly(L,T,q,q^\prime)2^{-0.5\kappa}$$
	Substitute it into (\ref{eq:74})(\ref{eq:94}), we get
	$$\fAdv^{naSymQKDM}_{\cF[q^\prime]}(\sA_{(L,T,q)},\kappa)\leq\fpoly(L,T,q,q^\prime)2^{-0.25\kappa}$$
\end{proof}

\end{document}

%% file: macros.tex
\usepackage{braket}
\usepackage{qcircuit}

\usepackage{amsmath,amssymb}
\usepackage{mathrsfs}
\usepackage{array}
\usepackage{makecell}
\usepackage{tikz}

\newtheorem{prtl}{Protocol}
\newtheorem{fact}{Fact}
\DeclareMathOperator{\tr}{tr}
\newcommand{\cR}{{\mathcal{R}}}
\newcommand{\cC}{{\mathcal{C}}}

\newcommand{\cD}{{\mathcal{D}}}
\newcommand{\cE}{{\mathcal{E}}}
\newcommand{\cF}{{\mathcal{F}}}
\newcommand{\cM}{{\mathcal{M}}}

\newcommand{\cH}{{\mathcal{H}}}
\newcommand{\cK}{{\mathcal{K}}}

\newcommand{\cS}{{\mathcal{S}}}
\newcommand{\cU}{{\mathcal{U}}}
\newcommand{\sA}{{\mathscr{A}}}

\newcommand{\mi}{{\mathrm{i}}}
\DeclareMathOperator{\bE}{\mathbb{E}}
\DeclareMathOperator{\bC}{\mathbb{C}}
\DeclareMathOperator{\bN}{\mathbb{N}}
\DeclareMathOperator{\bZ}{\mathbb{Z}}
\DeclareMathOperator{\bD}{\mathbb{D}}
\newcommand{\fH}{{\sf H}}
\newcommand{\fT}{{\sf T}}

\newcommand{\fX}{{\sf X}}
\newcommand{\fY}{{\sf Y}}
\newcommand{\fZ}{{\sf Z}}

\newcommand{\fI}{{\sf I}}
\newcommand{\fToffoli}{{\sf Toffoli}}

\newcommand{\fCN}{{\sf CNOT}}
\newcommand{\fSWAP}{{\sf SWAP}}
\newcommand{\fEn}{{\mathsf{Enc}}}
\newcommand{\fEt}{{\mathsf{Et}}}
\newcommand{\fDc}{{\mathsf{Dec}}}
\newcommand{\fEv}{{\mathsf{Eval}}}
\newcommand{\fKg}{{\mathsf{KeyGen}}}
\newcommand{\fVer}{{\mathsf{Ver}}}
\newcommand{\fneg}{{\mathsf{negl}}}
\newcommand{\fpoly}{{\mathsf{poly}}}

\newcommand{\fAdv}{{\mathsf{Adv}}}
\newcommand{\fBQC}{{\mathsf{BQC}}}

\newcommand{\fCL}{{\mathsf{CL}}}

\newcommand{\fTAB}{{\mathsf{TAB}}}
\newcommand{\fEvTAB}{{\mathsf{EvalTAB}}}
\newcommand{\fGBC}{{\mathsf{GBC}}}

\newcommand{\fKDMP}{{\mathsf{KDMP}}}

\newcommand{\fQKDM}{{\mathsf{QKDM}}}
\newcommand{\fOracle}{{\mathsf{Oracle}}}
\newcommand{\fShuffle}{{\mathsf{Shuffle}}}